\documentclass[journal,draftcls,onecolumn,12pt,twoside]{IEEEtranTCOM}
\DeclareMathAlphabet\mathbfcal{OMS}{cmsy}{b}{n}
\usepackage{cite,amsmath,amssymb,mathrsfs,epsf}
\usepackage{graphicx}
\usepackage[usenames]{color}
\usepackage{multirow}
\usepackage{tabularx}
\usepackage{stfloats}
\usepackage{verbatim}
\usepackage{epsfig}
\usepackage{amsfonts}
\usepackage[normalem]{ulem}
\usepackage{amsthm}
\usepackage[table]{xcolor}
\usepackage{float}
\usepackage{paralist}
\usepackage{hyperref}
\usepackage{color}
\usepackage{titlesec}
\usepackage{mathtools}
\usepackage[skip=0.5\baselineskip, font=small]{caption}
\usepackage[outdir=./]{epstopdf}
\IEEEoverridecommandlockouts

\normalsize

\allowdisplaybreaks

\usepackage[caption=false,font=footnotesize,subrefformat=parens,labelformat=parens]{subfig}
\makeatletter
\def\@IEEEsectpunct{:\ \,}
\def\paragraph{\@startsection{paragraph}{4}{\z@}{1.5ex plus 1.5ex minus 0.5ex}%
{0ex}{\normalfont\normalsize\itshape\bfseries}}
\makeatother
\DeclareMathOperator{\polylog}{polylog}

\begin{document}
%\title{The Random Arrival Bottleneck of Coded Caching}
%\title{Stochastic Coded Caching Networks}
%\title{Fundamental Limits of Stochastic Coded Caching Networks}
%\title{Communication and Load Balancing over Stochastic Caching Networks}
%\title{Fundamental Limits of Stochastic Caching Networks}
\title{Fundamental Limits of Stochastic Shared Caches Networks}
\author{\IEEEauthorblockN{Adeel~Malik, Berksan~Serbetci, Emanuele~Parrinello, Petros~Elia}
\thanks{The authors are with the Communication Systems Department at EURECOM, Sophia Antipolis, 06410, France (email: malik@eurecom.fr, serbetci@eurecom.fr, parrinel@eurecom.fr, elia@eurecom.fr). The work is supported by the European Research Council under the EU Horizon 2020 research and innovation program / ERC grant agreement no. 725929 (project DUALITY). This paper
was presented in part at the 2020 IEEE Global Communications (Globecom) Conference.}}


\newtheorem{axiom}{Axiom}
\newtheorem{lemma}{Lemma}
\newtheorem{corollary}{Corollary}
\newtheorem{theorem}{Theorem}
\newtheorem{prop}{Proposition}
\newtheorem{observation}{Observation}
\newtheorem{definition}{Definition}
\newtheorem{remark}{Remark}%[section]
\newtheoremstyle{case}{}{}{}{}{}{:}{ }{}
\theoremstyle{case}
\newtheorem{case}{Case}

\maketitle

\vspace*{-2cm}
\begin{abstract}
The work establishes the exact performance limits of stochastic coded caching when users share a bounded number of cache states, and when the association between users and caches, is random. Under the premise that more balanced user-to-cache associations perform better than unbalanced ones, our work provides a statistical analysis of the average performance of such networks, identifying in closed form, the exact optimal average delivery time. To insightfully capture this delay, we derive easy to compute closed-form analytical bounds that prove tight in the limit of a large number $\Lambda$ of cache states.  In the scenario where delivery involves $K$ users, we conclude that the multiplicative performance deterioration due to randomness --- as compared to the well-known deterministic uniform case --- can be unbounded and can scale as $\Theta\left( \frac{\log \Lambda}{\log \log \Lambda} \right)$ at $K=\Theta\left(\Lambda\right)$, and that this scaling vanishes when $K=\Omega\left(\Lambda\log \Lambda\right)$. To alleviate this adverse effect of cache-load imbalance, we consider various load balancing methods, and show that employing proximity-bounded load balancing with an ability to choose from $h$ neighboring caches, the aforementioned scaling reduces to~$\Theta \left(\frac{\log(\Lambda / h)}{ \log \log(\Lambda / h)} \right)$, while when the proximity constraint is removed, the scaling is of a much slower order $\Theta \left( \log \log \Lambda \right)$. The above analysis is extensively validated numerically.
\end{abstract}
\vspace*{-1cm}

\IEEEpeerreviewmaketitle

% \fi
\IEEEpeerreviewmaketitle

\setlength{\abovedisplayskip}{3pt}   %reduce the spacing between equation and the text line above
\setlength{\belowdisplayskip}{3pt}    %reduce the spacing between equation and the text line below
\setlength\abovecaptionskip{-0.1ex}
\setlength\belowcaptionskip{-5ex}
\setlength{\parskip}{-1pt}
\titlespacing*{\subsection}{0pt}{0.15\baselineskip}{0.15\baselineskip}
\titlespacing*{\section}{0pt}{0.2\baselineskip}{0.2\baselineskip}
\setlength{\skip\footins}{0.5cm}
\section{Introduction}\label{Sec:1}
\IEEEPARstart{E}{ver-increasing} volumes of mobile data traffic, have brought to the fore the need for new solutions that can serve a continuously increasing number of users, and do so with a limited amount of network bandwidth resources. In this context, cache-enabled wireless networks have emerged as a promising solution that can transform the storage capability of the nodes into a new and powerful network resource.

The potential of such cache-enabled wireless networks has been dramatically elevated following the seminal publication in~\cite{man_it14} which introduced the concept of \emph{coded caching}, and which revealed that --- in theory --- an unbounded number of users can be served even with a bounded amount of network resources. This was a consequence of a novel cache placement algorithm that enabled the delivery of independent content to many users at a time. Since then, several extensions of the basic coded caching setting have been studied. Such works include the study of coded caching for arbitrary file popularity distributions~\cite{man_it17, zhang_it18, ji_iswcs14}, various optimality results in \cite{yuTradeoff2TransIT2019,wanOptimalityTransIT2020,YuMA18}, results for various topology models~\cite{bidokhti_isit16, zhang_isit17, lampiris_isit17}, for MIMO broadcast channels~\cite{lampiris_jsac18,zhang_it17}, for PHY-based coded caching \cite{shariatpanahiMultiserverTransIT2016,tolli2017multi,lampiris_jsac18,shariatpanahi2018multi,ZFE:15,8007039,zhang_it17,lampiris_isit17,joudehMixedTrafficArXiv2020,cao_twc19}, for a variety of heterogeneous networks (HetNets)~\cite{hachem_it17,parrinello_it20}, D2D networks~\cite{ji_it16}, and other settings as well\cite{sengupta_it17,cao_twc17,roig_icc17,piovano_isit18,bayat_twc19,lampiris_isit18}.

\begin{figure}[t]
\centering
 \includegraphics[width=0.5\linewidth]{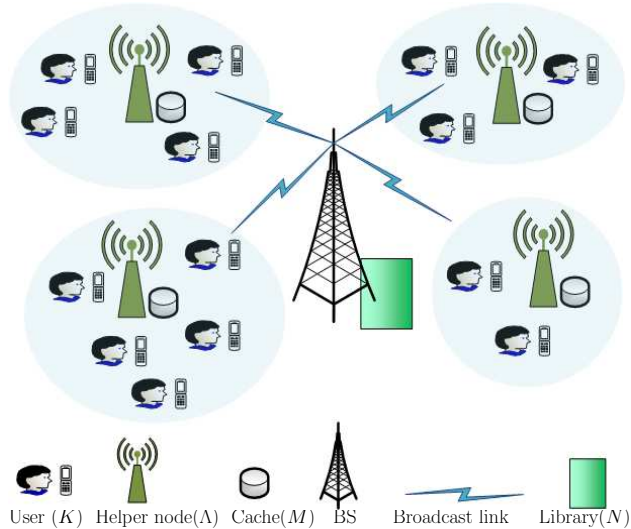}
\caption{An instance of a cache-aided heterogeneous network.}
\label{fig:SM}
\end{figure}

\subsection{Coded caching networks with shared caches}~\label{section:11}
Pivotal to the development of larger, realistic coded caching networks is the so-called \emph{shared-caches setting}, where different users are forced to benefit from the same cache content. This setting is of great importance because it reflects promising scenarios as well as unavoidable constraints. 

Such a promising scenario can be found in the context of cache-enabled heterogeneous networks, where a central transmitter (a base station) delivers content to a set of interfering users, with the assistance of cache-enabled helper nodes that serve as caches to the users. An instance of such a network is illustrated in Figure~\ref{fig:SM}. Such networks capture modern trends that envision a central base-station covering a larger area, in tandem with a multitude of smaller helper nodes each covering smaller cells. In this scenario, any user that appears in a particular small cell, can benefit from the cache-contents of the single helper node covering that cell. 

In the context of coded caching, an early work on this scenario can be found in~\cite{hachem_it17}, which employed the uniform user-to-cache association assumption where each helper node is associated to an equal number of users. This assumption was removed in~\cite{parrinello_it20}, which --- under the assumption that content cache placement is uncoded as well as agnostic to the user-to-cache association --- identified the exact optimal worst-case delivery time (i.e, the case of users' demand vector which requires the longest delivery time), as a function of the user-to-cache association profile that describes the number of users served by each cache.  A similar setting was studied in~\cite{karat_arxiv19} for the case of non-distinct requests, as well as in~\cite{wan_arxiv18, vu_wcnc17, parrinello_itw19} for the topology-aware (non-agnostic) scenario, where the user-to-cache association is known during cache placement. In this context, the work in~\cite{wan_arxiv18} proposed a novel coded placement that exploits knowledge of the user-to-cache association, while the work in~\cite{vu_wcnc17} used this same knowledge, to modulate cache-sizes across the different helper nodes as a function of how many users they serve. Similarly, the work in~\cite{parrinello_itw19} optimized over the cache sizes, again as a function of the load of each cache, and then proceeded to design a novel coded caching scheme which substantially outperforms the optimal scheme in  \cite{parrinello_it20}; the latter designed for the scenario where the cache placement is oblivious to the user-to-cache association phase. It is interesting to note that to a certain extent, this same shared-caches setting also applies to the scenario where each user requests multiple files (see for example~\cite{ji_gsip14, sengupta_tc17,wei_allerton17}).
 
Very importantly, this same shared-caches setting is directly related to the unavoidable subpacketization bottleneck because this bottleneck can force the use of a reduced number of distinct cache states that must be inevitably shared among the many users\footnote{When adopting the cache placement strategy in~\cite{man_it14} for $\Lambda$ caches, we refer the content to be placed in a single cache as a cache state.}. This number of distinct cache states, henceforth denoted as $\Lambda$, will be forced under most realistic assumptions, to be substantially less than the total number of users, simply because most known coded caching techniques require file sizes that scale exponentially with $\Lambda$ (see~\cite{shanmugam_it16,ji_icc15,yan_it17,tang_isit17,shangguan_it18,shanmugam_isit17,lampiris_jsac18}). In a hypothetical scenario where coded caching is applied across a city of, let's say, one million mobile users, one would have to assign each user with one of $\Lambda$ cache states ($\Lambda$ independent caches), where $\Lambda$ would probably be forced to be in the double or perhaps triple digits~\cite{shanmugam_it16}. 

As one can see, both of the above isomorphic settings imply that during the content delivery that follows the allocation of cache-states to each user, different broadcast sessions would experience user populations that differently span the spectrum of cache states. In the most fortunate of scenarios, a transmitter would have to deliver to a set of $K$ users that uniformly span the $\Lambda$ states (such that each cache state is found in exactly $K/\Lambda$ users), while in the most unfortunate of scenarios, a transmitter would encounter $K$ users that happen to have an identical cache state. Both cases are rare instances of a stochastic process, which we explore here in order to identify the exact optimal performance of such systems.

Most of our results apply both to the heterogeneous network scenario as well as the aforementioned related subpacketization-constrained setting which was nicely studied in~\cite{jin_toc_19}. This interesting work in~\cite{jin_toc_19} introduced the main problem, and after providing upper bounds\footnote{It is worth noting that the provided upper bounds in~\cite{jin_toc_19} can have a large gap from the here-derived optimal.}, introduced the challenge of identifying the fundamental limits of this same problem. This is indeed the challenge that is resolved here, where we are able to derive these fundamental limits of performance, in their exact form. The problem for which we are deriving these fundamental limits, is itself a crucial problem in the context of coded caching. Based on our analysis, and our new insight that indeed we can have an unbounded ``damage" from the stochastic nature of the user-to-cache association, we proceed to provide simple and realistic load balancing techniques. Our analysis captures the benefits of the techniques, in their exact form.

For ease of exposition, we will focus the wording of our description to the first scenario corresponding to a heterogeneous network where $\Lambda$ plays the role of the number of helper nodes. All the results though of Section~\ref{Sec:2} certainly apply to the latter setting as well.

\subsection{Load balancing in mobile networks}~\label{Sec:10}
As mentioned above, we will be focusing on heterogeneous networks and will explore the statistical properties brought about by the introduction of cache-aided helper nodes. We have also suggested that performance generally suffers when the different helper nodes are unevenly loaded. For this, it is only natural that we look at basic load balancing approaches, which have long played a pivotal role in improving the statistical behavior of wireless networks. This role was highlighted in the survey found in~\cite{Andrews2014}, which discussed why long-standing assumptions about cellular networks need to be rethought in the context of load balanced heterogeneous networks, and showed that natural user association metrics like signal-to-interference-plus-noise ratio (SINR) or received signal strength indication (RSSI) can lead to a major imbalance. This work gathered together the earlier works on load balancing in HetNets and compared the primary technical approaches -- such as optimization, game theory and Markov decision processes -- to HetNet load balancing.  In the same context, various algorithms have also been proposed to optimize the traffic load by analyzing user association to servers for cellular networks~\cite{Ye2013}, by providing a distributed $\alpha-$optimal user association and cell load balancing algorithm for wireless networks~\cite{Kim2012}, by developing SINR-based flexible cell association policies in heterogeneous networks~\cite{Jo2012}, and even investigating traffic load balancing in backhaul-constrained cache-enabled small cell networks powered by hybrid energy sources~\cite{Han2017}.

In this paper, we build a bridge between load balancing and coded caching, with the aim of improving the network performance by balancing the user load placed on each cache. We will show that the effect of load balancing can in fact be unbounded in the limit of many caches. 

\subsection{Shared-caches setting \& problem statement}~\label{Sec:11}
We consider the shared-caches coded-caching setting where a transmitter having access to a library of $N$ equisized files, delivers content via a broadcast link to $K$ receiving users, with the assistance of $\Lambda$ cache-enabled helper nodes. Each helper node $\lambda \in \left[1, 2, \dots,\Lambda \right]$ is equipped with a cache of storage capacity equal to the size of $M$ files, thus being able to store a fraction $\gamma \triangleq \frac{M}{N}\in \left[\frac{1}{\Lambda}, \frac{2}{\Lambda}, \dots,1 \right]$ of the library.
Each such helper node, which will be henceforth referred to as a `cache', can assist in the delivery of content to any number of receiving users.

The communication process consists of three phases; the \emph{content placement phase}, the \emph{user-to-cache association phase}, and the \emph{delivery phase}. The first phase involves the placement of library-content in the caches, and it is  oblivious to the outcome of the next two phases. The second phase is when each user is assigned -- independently of the placement phase -- to exactly one cache from which it can download content at zero cost. This second phase is also oblivious of the other two phases\footnote{This assumption is directly motivated by the time-scales of the problem, as well as by the fact that in the heterogeneous setting, the user-to-cache association is a function of the geographical location of the user. Note that users can only be associated to caches when users are within the coverage of caches, and a dynamic user-to-cache association that requires continuous communication between the users and the server may not be desirable as one seeks to minimize the network load overhead and avoid the handover.}. The final phase begins with users simultaneously requesting one file each, and continues with the transmitter delivering this content to the receivers. Naturally this phase is aware of the content of the caches, as well as aware of which cache assists each user.

\paragraph*{User-to-cache association}
For any cache $\lambda\in\ \left[1, \dots,\Lambda \right]$, we denote by $v_{\lambda}$ the number of users that are assisted by it, and we consider the \emph{cache population vector} $\mathbf{V}= \left[v_1,  \dots, v_{\Lambda} \right]$. Additionally we consider the sorted version $\mathbf{L}= \left[l_1,  \dots, l_{\Lambda} \right] = sort(\mathbf{V})$, where $sort(\mathbf{V})$ denotes the sorting of vector $\mathbf{V}$ in descending order. We refer to $\mathbf{L}$ as a \emph{profile vector}, and we note that each entry $l_{\lambda}$ is simply the number of users assisted by the $\lambda$-th most populous (most heavily loaded) cache. Figure~\ref{fig:SM} depicts an instance of our shared-caches setting where $\mathbf{L}= \left[5, 4, 3, 2\right]$.

\paragraph*{Delivery phase}
The delivery phase commences with each user $k \in \left[1, 2, \dots, K\right]$ requesting a single library file that is indexed by $d_k \in \left[1, 2, \dots, N\right]$. As is common in coded caching works, we assume that each user requests a different file. Once the transmitter is notified of the \emph{request vector} $\mathbf{d}=\left[d_1, d_2, \dots, d_K \right]$, it commences delivery over an error-free broadcast link of bounded capacity per unit of time.
%of one file per time slot.

\subsection{Metric of interest}\label{sec:metricofinterest}
As one can imagine, any given instance of the problem, experiences a different user-to-cache association, and thus\footnote{We briefly note that focusing on $\mathbf{V}$ rather than the sets of users connected to each cache, maintains all the pertinent information, as what matters for the analysis is the number of users connected to each cache and not the index (identity) of the users connected to that cache.} a different $\mathbf{V}$. Our measure of interest is thus the average delay
\begin{equation}
\overline{T}(\gamma) \triangleq E_{\mathbf{V}}[T(\mathbf{V})] = \sum_{\mathbf{V}}P(\mathbf{V})T(\mathbf{V}),
\end{equation}
where $T(\mathbf{V})$ is the worst-case delivery time\footnote{This delay corresponds to the time needed to complete the delivery of any file-request vector $\mathbf{d}$, where the time scale is normalized such that a unit of time corresponds to the optimal amount of time needed to send a single file from the transmitter to the receiver, had there been no caching and no interference.} corresponding to any specific cache population vector $\mathbf{V}$, and where $P(\mathbf{V})$ is the probability that the user-to-cache association corresponds to vector $\mathbf{V}$. 

More precisely, we use $T(\mathbf{V}, \mathbf{d}, \mathcal{X})$ to define the delivery time required by some generic caching-and-delivery scheme $\mathcal{X}$ to satisfy request vector $\mathbf{d}$ when the user-to-cache association is described by the vector $\mathbf{V}$. Our aim here is to characterize the optimal average delay
\begin{align}
\overline{T}^*(\gamma) & = \min_{\mathcal{X}} E_{\mathbf{V}}\left[ \max_{\mathbf{d}}T(\mathbf{V},\mathbf{d},\mathcal{X})\right]  = \min_{\mathcal{X}} E_{\mathbf{L}}\left[
                     E_{\mathbf{V}_{\mathbf{L}}}\left[  \max_{\mathbf{d}}T(\mathbf{\mathbf{V},\mathbf{d}},\mathcal{X})\right]\right],\label{eq:DefMetric} 
\end{align}
where the minimization is over all possible caching and delivery schemes $\mathcal{X}$, and where $E_{\mathbf{V_{\mathbf{L}}}}$ denotes the expectation over all vectors $\mathbf{V}$ whose sorted version is equal to some fixed  $sort(\mathbf{V})=\mathbf{L}$. Consequently the metric of interest takes the form
\begin{equation}
\overline{T}(\gamma) =  E_{\mathbf{L}}[T(\mathbf{L})] = \sum_{\mathbf{L}}P(\mathbf{L})T(\mathbf{L}),
\end{equation}
where $T(\mathbf{L}) \triangleq E_{\mathbf{V}_{\mathbf{L}}}\left[  \max_{\mathbf{d}}T(\mathbf{\mathbf{V},\mathbf{d}})\right]$, and where \[P(\mathbf{L})  \triangleq \!\!\!\! \!\!\sum_{\mathbf{V} :  sort(\mathbf{V})= \mathbf{L} } \!\!\!\!\! P(\mathbf{V}),\]  is simply the cumulative probability over all $\mathbf{V}$ for which $sort(\mathbf{V}) = \mathbf{L}$.

We will consider here the uncoded cache placement scheme in~\cite{man_it14}, and the delivery scheme in~\cite{jin_toc_19,parrinello_it20}, which will prove to be optimal for our setting under the common assumption of uncoded cache placement. This multi-round delivery scheme introduces --- for any $\mathbf{V}$ such that $sort(\mathbf{V}) = \mathbf{L} $ --- a worst-case delivery time of
\begin{align}\label{eq:TL}
 T(\mathbf{L})  =    \sum_{\lambda=1}^{\Lambda-t} l_{\lambda} \frac{{\Lambda-\lambda \choose t} }{{\Lambda \choose t}},
 \end{align}
\noindent where $t=\Lambda\gamma$.
%, \BS{and $l_\lambda$ is the number of users in the $\lambda$-th most populous cache.}

From equation \eqref{eq:TL} we can see that the minimum delay corresponds to the case when $\mathbf{L} $ is uniform. When $\Lambda$ divides $K$, this minimum (uniform) delay takes the well-known form
\begin{align}\label{eq:tbc0}
T_{min} = \frac{K(1-\gamma)}{1+\Lambda\gamma},
\end{align}
while for general $K,\Lambda$, it takes the form\footnote{When $K/\Lambda\notin\mathbb{Z}^+$, the best-case delay corresponds to having $l_{\lambda} = \left\lfloor K/\Lambda\right\rfloor +1$ for $\lambda \in \left[1, 2, \cdots, \hat{K}\right]$ and $l_{\lambda}= \left\lfloor K /\Lambda\right\rfloor$  for  $\lambda \in \left[\hat{K}\!+1, \hat{K}\!+2, \cdots, \Lambda\right]$, where $\hat{K}=K- \left\lfloor K/\Lambda\right\rfloor\Lambda$.}
\begin{align}\label{eq:tbc}
T_{min} = \frac{\Lambda-t}{1+t}  \left( \left\lfloor\frac{K}{\Lambda}\right\rfloor +1 - f(\hat{K})\right),
\end{align}
where $\hat{K}=K- \left\lfloor \frac{K}{\Lambda} \right\rfloor \Lambda$, $f(\hat{K})=1$ when $\hat{K}=0$, $f(\hat{K})=0$ when $\hat{K} \geq \Lambda-t$, and $f(\hat{K})= \frac{\prod_{i=t+1}^{\hat{K}+t} (\Lambda - i)}{ \prod_{j=0}^{\hat{K}-1} (\Lambda - j)}$ when $\hat{K} < \Lambda-t$.
The proof of this is straightforward, but for completeness it can also be found in Appendix~\ref{AP:tmin}. 
The above $T_{min}$ is optimal under the assumption of uncoded placement (cf. \cite{parrinello_it20}). 

On the other hand, for any other (now non-uniform) $\mathbf{L}$, the associated delay $T(\mathbf{L})$ will exceed $T_{min}$ (see \cite{parrinello_it20} for the proof, and see~Figure~\ref{fig:loss} for a few characteristic examples), and thus so will the average delay
\begin{figure}[t]
\centering
\includegraphics[ width=.99\linewidth]{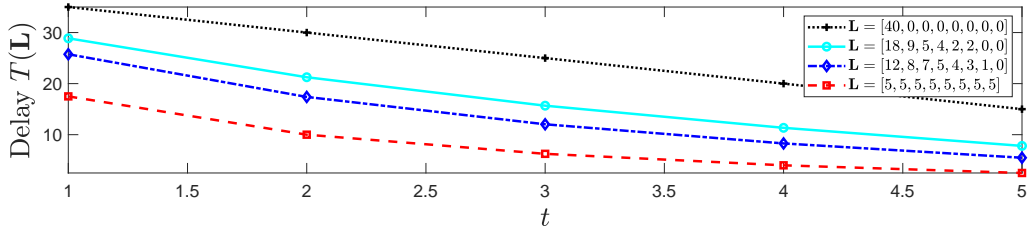}
\caption{Delay $T(\mathbf{L})$ for different profile vectors $\mathbf{L}$, for $K\!\!=\!\!40$ and $\Lambda\!=\!8$.}
\label{fig:loss}
\end{figure}
\begin{align}\label{eq:tavg}
E_{\mathbf{L}}[T(\mathbf{L})] &= \sum_{\mathbf{L} \in \mathcal{L} }P(\mathbf{L})  T(\mathbf{L}) =   \sum_{\lambda=1}^{\Lambda-t} \sum_{\mathbf{L} \in \mathcal{L} }P(\mathbf{L}) l_{\lambda}     \frac{ {\Lambda-\lambda \choose t} }{{\Lambda \choose t}}  =   \sum_{\lambda=1}^{\Lambda-t} E[ l_{\lambda} ]   \frac{ {\Lambda-\lambda \choose t} }{{\Lambda \choose t}},
 \end{align}
\noindent where $\mathcal{L}$ describes the set of all possible profile vectors $\mathbf{L}$ (where naturally $\sum_{\lambda=1}^{\Lambda}  l_{\lambda} = K$), and where $E[l_{\lambda} ]$ is the expected number of users in the $\lambda$-th most populous cache\footnote{It is straightforward to see that $\sum_{\mathbf{L} \in \mathcal{L}} l_{\lambda} P(\mathbf{L})$ is equivalent to $\sum_{j=0}^{K} j P(l_{\lambda}= j) = E[ l_{\lambda} ]$, where $P(l_{\lambda}= j)= \sum_{{\mathbf{L} \in \mathcal{L} : \mathbf{L}(\lambda)=j  }}  P(\mathbf{L})$. }.

\subsection{Our contribution}~\label{section:14}
%In this work we assume that each user can be associated to any particular cache (i.e., can appear in any particular cell) with equal probability.
In this work we assume that each user can appear in the coverage area of any particular cache-enabled cell (i.e., can be associated to any particular cache) with equal probability. We will identify the optimal average delay $\overline{T}^*(\gamma)$ and the corresponding (multiplicative) performance deterioration
\begin{align}\label{eq:G}
  G(\gamma)= \frac{\overline{T}^*(\gamma)}{T_{min}}
\end{align}
experienced in this random setting. Our aim is to additionally provide expressions that can either be evaluated in a numerically tractable way, or that can be asymptotically approximated in order to yield clear insight.  The following are our contributions, step by step.
\begin{itemize}
\item  In Section~\ref{Sec:21}, we characterize in closed form the exact optimal average delay $\overline{T}^*(\gamma)$, optimized over all placement and delivery schemes under the assumption of uncoded cache placement and under the assumption that each user can appear in the coverage area of any particular cache-enabled cell with equal probability.
%can be associated to any particular cache with equal probability.  
\item To simplify the numerical interpretation of the above expression, we propose in Section~\ref{Sec:22} analytical bounds that can be calculated efficiently.  
\item In Section~\ref{Sec:23}, we characterize the exact scaling laws of performance. It is interesting to see that the aforementioned multiplicative deterioration $G(\gamma)= \frac{\overline{T}^*(\gamma)}{T_{min}}$ can in fact be unbounded, as $\Lambda$ increases. For example, when $K=\Theta\left(\Lambda\right)$ (i.e., when $K$ matches the order of $\Lambda$), the performance deterioration scales exactly as $\Theta\left( \frac{\log\Lambda}{\log \log\Lambda} \right)$, whereas when $K$ increases, this deterioration gradually reduces, and ceases to scale when $K=\Omega\left(\Lambda\log\Lambda\right)$.
\item In Section~\ref{Sec:3}, we use two load balancing approaches to alleviate the effect of randomness.  In the practical scenario where we are given a choice to associate a user to the least loaded cache from a randomly chosen group of $h$ \emph{neighboring} helper nodes, the performance deterioration stops scaling as early as $K=\Omega\left(\frac{\Lambda}{h}\log\frac{\Lambda}{h}\right)$. An even more dramatic improvement can be seen when the aforementioned neighboring/proximity constraint is lifted.
The above reveals that load balancing, when applicable, can play a crucial role in significantly reducing the performance deterioration due to random user-to-cache association.
\item In Section~\ref{Sec:4}, we perform extensive numerical evaluations that validate our analysis.
\item In Section~\ref{Sec:NonUniform}, we extend our analysis to the scenario where cache population intensities (i.e, probability that a user can appear in the coverage area of any particular cache-enabled cell) are following a non-uniform distribution.
\end{itemize}
\subsection{Notations}~\label{section:16}
Throughout this paper, we use $[x]\triangleq\left[1, 2, \dots, x \right]$, and we use $\mathbf{A} / \mathbf{B}$ to denote the difference set that consists of all the elements of set $\mathbf{A}$ not in set $\mathbf{B}$. Unless otherwise stated, logarithms are assumed to have base 2. We use the following asymptotic notation: i) $f(x)= O(g(x))$ means that there exist constants $a$ and $c$ such that $f(x) \leq ag(x), \forall x > c$, ii) $f(x) = o(g(x))$ means that  $\lim_{x \rightarrow \infty  }\frac{f(x)}{g(x)} =0 $, iii) $f(x)=\Omega(g(x))$ if $g(x)= O(f(x))$, iv) $f(x) = \omega(g(x))$ means that  $\lim_{x \rightarrow \infty  }\frac{g(x)}{f(x)} =0 $, v)  $f(x)=\Theta(g(x))$ if $f(x)= O(g(x))$  and $f(x)=\Omega(g(x))$.  We use the term $\polylog(x)$ to denote the class of functions $\bigcup_{k\geq 1} O((\log x)^k)$ that are polynomial in $\log x$. 

\section{Main Results}\label{Sec:2}
In this section we present our main results on the performance of the $K$-user broadcast channel with $\Lambda$ caches, each of normalized size $\gamma$, and a uniformly random user-to-cache association process. As noted, the analysis applies both to the $\Lambda$-cell heterogeneous network, as well as to the isomorphic subpacketization-constrained setting.    

\subsection{Exact characterization of the optimal average delay}\label{Sec:21}
We proceed to characterize the exact optimal average delay $\overline{T}^*(\gamma)$. Crucial in this characterization will be the vector $\mathbf{B}_{\mathbf{L}}= \left[b_1, b_2, \dots, b_{\left|B_{\mathbf{L}}\right|} \right]$, where each element $b_j \in \mathbf{B}_{\mathbf{L}}$ indicates the number of caches in a distinct group of caches in which each cache has the same load\footnote{For example, for a  profile vector $\mathbf{L}=\left[5,5,3,3,3,2,1,0,0\right]$, there are five distinct groups in terms of having the same load, then the corresponding vector $\mathbf{B}_{\mathbf{L}}=\left[2,3,1,1,2 \right]$, because two caches have a similar load of five users, three caches have a similar load of three users, two caches have a similar load of zero and all other caches have distinct number of users.}. Under the assumption that each user can be associated to any particular cache with equal probability, the optimal average delay $\overline{T}^*(\gamma)$ --- optimized over all coded caching strategies with uncoded placement --- is given by the following theorem.  
\begin{theorem}\label{th:egap}
In the $K$-user, $\Lambda$-caches setting with normalized cache size $\gamma$ and a random user-to-cache association, the average delay 
\begin{align} \label{eq:etavg}
\overline{T}^*(\gamma)= \sum_{\lambda=1}^{\Lambda-t} \sum_{\mathbf{L} \in \mathcal{L}} \frac{K! \ t! \ (\Lambda-t)! \ l_{\lambda}  {\Lambda-\lambda \choose t}}{ \Lambda^{K}\prod _{i=1}^{\Lambda}l_i!\  \prod _{j=1}^{\left|\mathbf{B}_{\mathbf{L}}\right|}b_j! }
\end{align}
is exactly optimal under the assumption of uncoded placement.
\end{theorem} 
\begin{proof} The proof can be found in~Appendix~\ref{AP:egap}. 
\end{proof}
One can now easily see that when $\frac{K}{\Lambda}\in \mathbb{Z}^{+}$, the optimal multiplicative deterioration 
$G(\gamma) = \frac{\overline{T}^*(\gamma)}{T_{min}}$ takes the form
\begin{align} \label{eq:gap}
G(\gamma)= \sum_{\lambda=1}^{\Lambda-t} \sum_{\mathbf{L} \in \mathcal{L}} \frac{ (K-1)! \  (\Lambda-t-1)! \  (t+1)! \  l_{\lambda}  {\Lambda-\lambda \choose t}}{\Lambda^{K-1} \  \prod _{i=1}^{\Lambda}l_i! \  \prod _{j=1}^{\left|\mathbf{B}_{\mathbf{L}}\right|}b_j! }. \end{align}
% \begin{table}[t]
% \centering
% \begin{tabular}{|c|c|c|}
% \hline
%      & \multicolumn{1}{c|}{\begin{tabular}[c]{@{}c@{}}$\vert\mathcal{L}\vert$\end{tabular}} & \multicolumn{1}{c|}{\begin{tabular}[c]{@{}c@{}}$ \vert\mathcal{V}\vert$\end{tabular}} \\ \hline
% $K=10$ &          $42$      &   $92378$        \\ \hline
% $K=20$ &     $530$     &   $10015005$        \\ \hline
% $K=30$ &              $3590$   &   $211915132$    \\ \hline
% $K=40$ &        $16928$  &    $2.054455634\times10^9$ \\ \hline
% $K=50$ &    $62740$        &     $1.2565671261\times10^{10}$  \\ \hline
% \end{tabular}
% \caption{Size of $\mathcal{L}$ and $\mathcal{V}$ ($\Lambda = 10$)}
% \label{tab:LvsV}
% \end{table}

\begin{table}[t]
\centering
\begin{tabular}{|c|c|c|c|c|c|}
\hline
\multicolumn{1}{|l|}{} & $K=10$ & $K=20$ &$K=30$  & $K=40$ & $K=50$ \\ \hline
 $\vert\mathcal{V}\vert$                      &$92378$   & $10015005$  & $211915132$ & $2.054455634\times10^9$ & $1.2565671261\times10^{10}$ \\ \hline
   $ \vert\mathcal{L}\vert$                    &$42$  &$530$  & $3590$ &  $16928$  &  $62740$  \\ \hline
\end{tabular}
\caption{Size of $\mathcal{L}$ and $\mathcal{V}$ ($\Lambda = 10$)}
\label{tab:LvsV}
\end{table}

\begin{remark}~\label{R:1}
Theorem~\ref{th:egap} provides the exact optimal performance in the random
association setting, as well as a more efficient way to evaluate this
performance compared to the state of the art (SoA) (cf.~\cite[Theorem 1]{jin_toc_19}). The worst-case computational time complexity for calculating the exact optimal average delay $\overline{T}^*(\gamma)$ is $O\left(\max(K, \left|\mathcal{L}\right|\Lambda ) \right)$ as compared to the $O\left(\max(K, \left|\mathcal{V}\right|\Lambda \log \Lambda) \right)$ for the case of~\cite[Theorem 1]{jin_toc_19}. This speedup is due to the averaging being over the much smaller set $\mathcal{L}$ of all $\mathbf{L}$, rather than over the set $\mathcal{V}$ of all $\mathbf{V}$ (see Table~\ref{tab:LvsV} for a brief comparison). The time complexities mentioned above do not include the cost of creating the sets $\mathcal{L}$ and $\mathcal{V}$. However, we note that the creation of $\mathcal{V}$ is a so-called weak composition problem, whereas the creation of $\mathcal{L}$ is an integer partition problem~\cite{partsize}. It is easy to verify that the complexities of the algorithms for the integer partition problem are significantly lower than the ones for the weak composition problem~\cite{stojmenovic_ijcm98,vajnovszki_arxiv13,merca_jmma12,kelleher_arxiv09}.
\end{remark}

Despite the aforementioned speedup, exact evaluation of~\eqref{eq:etavg} can still be computationally expensive for large parameters. This motivates our derivation of much-faster to evaluate analytical bounds on $\overline{T}^*(\gamma)$, which we provide next.
\subsection{Computationally efficient bounds on the optimal performance}\label{Sec:22}
The following theorem bounds the optimal average delay $\overline{T}^*(\gamma)$.
\begin{theorem}\label{th:UBLB}
In the $K$-user, $\Lambda$-cache setting with normalized cache size $\gamma$ and a random user-to-cache association, the optimal average delay $\overline{T}^*(\gamma)$ is bounded by 
%\begin{align}\label{eq:UBtavg}
%\overline{T}^*(\gamma) \leq K\frac{ \Lambda-t }{ t+1} - \sum_{\lambda=1}^{\Lambda-t} \frac{ {\Lambda-\lambda \choose t} }{{\Lambda \choose t}}  \sum_{j=0}^{K-1}  \max\left(1-\frac{\Lambda}{\lambda}(1-P_j), 0\right) 
%\end{align}
\begin{align}\label{eq:UBtavg}
\overline{T}^*(\gamma) \leq K\frac{ \Lambda-t }{ t+1} - \sum_{\lambda=1}^{\Lambda-t} \frac{ {\Lambda-\lambda \choose t} }{{\Lambda \choose t}} \sum_{j=0}^{K-1}  \max\left(1-\frac{\Lambda}{\lambda}(1-P_j), 0\right),
\end{align}
and 
%\begin{align} \label{eq:LBtavg}
%\overline{T}^*(\gamma)\geq \frac{\Lambda-t}{1+t} \left( \frac{K}{\Lambda}  \frac{  (\Lambda-t-1)}{(\Lambda-1)}+\frac{ t  }{ (\Lambda-1)}\left(K- \sum_{j=\left \lceil \frac{K}{\Lambda} \right\rceil}^{K-1}  P_j \right) \right)
% \end{align}
 \begin{align} \label{eq:LBtavg}
\overline{T}^*(\gamma)\geq \frac{\Lambda-t}{1+t}\left(\frac{K}{\Lambda}  \frac{\Lambda-t-1}{\Lambda-1}+ \frac{t}{\Lambda-1}\left(K- \sum_{j=\left \lceil \frac{K}{\Lambda} \right\rceil}^{K-1} P_j \right)  \right),
 \end{align}
\noindent where 
\begin{align}\label{eq:cdft}
P_j =\sum_{i=0}^{j} {K \choose i} \left(\frac{1}{\Lambda}\right)^i \left(1-\frac{1}{\Lambda}\right)^{K-i}. 
 \end{align}
\end{theorem}
\begin{proof} 
The proof is deferred to Appendix~\ref{AP:UBLB}. 
% Crucial to this proof is the exploitation of the fact that $\sum_{\lambda=1}^{\Lambda} E[ l_{\lambda} ]= K$, of the fact that both $E[ l_{\lambda}]$ and ${\Lambda-\lambda \choose t}$ in \eqref{eq:tavg} are non-increasing with $\lambda$, and of the fact that $E[l_{1}]$ was evaluated in closed form. 
\end{proof}
\begin{figure}[t]
\centering
\includegraphics[ width=1\linewidth]{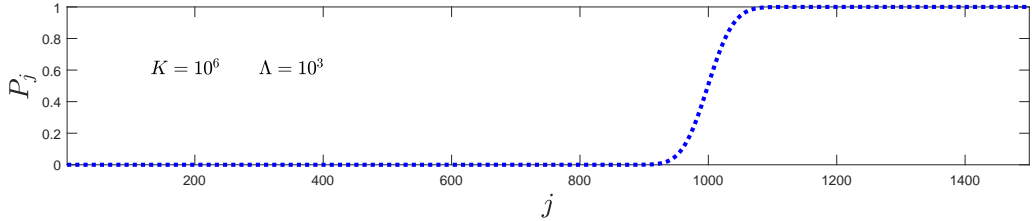} 
\caption{Behavior of $P_j$ for $K = 10^6$ and $\Lambda = 10^3$.}
\label{fig:cdf}
\end{figure}
\begin{remark}~\label{R:2}
The worst-case computational time complexity for calculating the analytical bounds on the optimal performance $\overline{T}^*(\gamma)$ based on Theorem \ref{th:UBLB} is $O\left( \max(K\log K,  K\Lambda) \right)$. This is significantly better compared to the the complexity of $O\left(\max(K, \left|\mathcal{L}\right|\Lambda )\right)$ for the exact calculation (cf. Theorem~\ref{th:egap}) The above bound is computationally efficient due to its dependence only on the $P_j$ (cf.~\eqref{eq:cdft}), which is the cumulative distribution function (cdf) of a random variable that follows the binomial distribution with $K$ independent trials and $\frac{1}{\Lambda}$ success probability. To compute bounds, the value of $P_j$ needs to be calculated for all values of $j\in[0,K-1]$, which can be computationally expensive (i.e., $O\left( K\log K\right)$). However, as is known, there exists a $\tilde{j} \in [0,K-1]$, where $P_j \approx 1$. Since the cdf is a non-decreasing function in $j$, it is clear\footnote{The well-known De Moivre-Laplace Theorem can help us gain some intuition as to why the above method is computationally efficient and precise. In our case here, our binomial distribution --- which according to the aforementioned theorem can be approximated by the normal distribution in the limit of large $K$ --- has mean $K/\Lambda$ and standard deviation $\sqrt{K(\Lambda-1)/\Lambda^2}$. This simply means that the values within three standard deviations of the mean account for about $99.7\% \approx 100\%$ of the set. This in turn means that $P_{\tilde{j}} \approx 1$ as early on as $\tilde{j} = K/\Lambda + 3\sqrt{K(\Lambda-1)/\Lambda^2} << K$. Since $P_j \approx 1$ for $j \geq \tilde{j}$, implies that \eqref{eq:cdft} can be rapidly evaluated with high precision.} that $P_j \approx 1$ for $j>\tilde{j}$. An illustration for $K = 10^6$, and $\Lambda = 10^3$ is shown in Figure~\ref{fig:cdf}, where it is evident that $\tilde{j} << K$. 
\end{remark}

Directly from Theorem~\ref{th:UBLB} and equation~\eqref{eq:tbc0}, we can conclude that for $\frac{K}{\Lambda}\in \mathbb{Z}^{+}$, the performance deterioration $G(\gamma)$ as compared to the deterministic uniform case, is bounded as
\begin{align}\label{eq:UB}
G(\gamma) &\leq \Lambda - \frac{t+1}{K-K\gamma}\sum_{\lambda=1}^{\Lambda-t} \frac{ {\Lambda-\lambda \choose t} }{{\Lambda \choose t}}\sum_{j=0}^{K-1}  \max\left(1-\frac{\Lambda}{\lambda}(1-P_j), 0\right),
\end{align}
and 
 \begin{align} \label{eq:LB}
G(\gamma) \geq  \frac{\Lambda-t-1}{\Lambda-1}+ \frac{\Lambda}{K}\frac{t}{\Lambda-1}\left(K-\sum_{j=\left \lceil \frac{K}{\Lambda} \right\rceil}^{K-1}P_j\right),
 \end{align}
\noindent where $P_j$ is given in Theorem~\ref{th:UBLB}.
 
We now proceed to provide the exact scaling laws of the fundamental limits of the performance in a simple and insightful form.
\subsection{Scaling laws of coded caching with random association}\label{Sec:23}
The following theorem provides the asymptotic analysis of the optimal $\overline{T}^*(\gamma)$, in the limit of large $\Lambda$.
\begin{theorem}\label{th:gap}
In the $K$-user, $\Lambda$-caches setting with normalized cache size $\gamma$ and random user-to-cache association, the optimal delay scales as 
\begin{align} \label{eq:lawtavg}
\!\!\overline{T}^*(\gamma) \!= \!
\begin{cases} 
\!\Theta\! \left(\frac{T_{min}\Lambda\log\Lambda}{K\log\frac{\Lambda \log\Lambda }{K}} \right) \text{\emph{if}  } K \in \left[\frac{\Lambda}{\polylog(\Lambda)}, o\left(\!\Lambda \!\log\Lambda\! \right)\right]\\
\! \Theta\left(T_{min}\right) \hspace{0.85cm} \text{ \emph{if} }  K=\Omega\left(\Lambda \log\Lambda \right). 
\end{cases}
\end{align} 
\end{theorem}
\begin{proof} 
Deferred to Appendix~\ref{AP:gap}.
\end{proof} 
Directly from the above, we now know that the performance deterioration due to user-to-cache association randomness, scales as
\begin{align} \label{eq:egap}
\!\!\!G(\gamma)\!= \!
\begin{cases} 
\!\Theta\! \left(\!\frac{\Lambda\log\Lambda}{K\log\frac{\Lambda \log\Lambda }{K}} \!\!\right)\! \text{ if }  K \in  \left[ \frac{\Lambda}{\polylog(\Lambda)},o\left(\!\Lambda \!\log\Lambda\! \right)\right]\\
\! \Theta\left(1\right) \hspace{1.25cm} \text{ if }  K=\Omega\left(\Lambda \log\Lambda \right),
\end{cases}
\end{align}
which in turn leads to the following corollary. 
\begin{corollary}~\label{c:gap}
The performance deterioration $G(\gamma)$ due to association randomness, scales as $\Theta\left( \frac{\log \Lambda}{\log \log \Lambda} \right)$ at $K=\Theta\left(\Lambda\right)$, and as $K$ increases, this deterioration gradually reduces, and ceases to scale when $K=\Omega\left(\Lambda\log \Lambda\right)$.
\end{corollary}
\begin{proof}
The proof is straightforward from Theorem~\ref{th:gap}.
\end{proof}
In identifying the exact scaling laws of the problem, Theorem~\ref{th:gap} nicely captures the following points.
\begin{itemize}
    \item It describes the extent to which the performance deterioration increases with $\Lambda$ and decreases with $\frac{K}{\Lambda}$.
    \item It reveals that the performance deterioration can in fact be unbounded. 
    \item It shows how in certain cases, increasing $\Lambda$ may yield diminishing returns due to the associated exacerbation of the random association problem. For example, to avoid a scaling $G(\gamma)$, one must approximately keep $\Lambda$ below $e^{W(K)}$ ($W(.)$ is the Lambert W-function) such that $\Lambda\log \Lambda\leq K$.
\end{itemize}    
The detrimental impact of the user-to-cache association's randomness on the delivery time motivates the need of techniques to mitigate this impact. In Section \ref{Sec:3}, we show how incorporating load balancing methods in shared cache setting can play a vital role in mitigating this impact.

\begin{figure}[t]
\centering
\includegraphics[ width=1\linewidth]{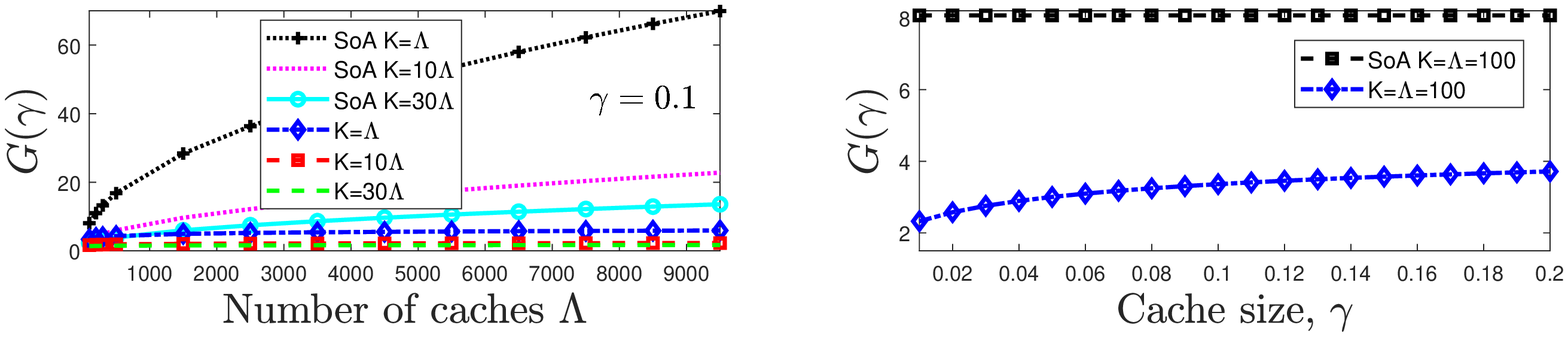} 
\caption{Upper bound comparison with SoA.}
\label{fig:soa}
\end{figure}

\subsection{Furthering the SoA on the subpacketization-constrained decentralized coded
caching setting}\label{Sec:5}
As mentioned before, in general the shared cache setting is isomorphic to the subpacketization-constrained coded caching setting~\cite{shanmugam_it16}, where each cache-enabled user is forced to store the content from one of $\Lambda$ cache states. In particular, the work in~\cite{jin_toc_19} proposed a decentralized coded caching in this subpacketization-constrained setting, where each cache-enabled user stores the content from one of $\Lambda$ cache states with equal probability, which is exactly equivalent to our setting where each user appears in the coverage area of any particular cache-enabled cell with equal probability. We briefly mention below the utility of our results in this latter context. 
\begin{itemize}
\item Theorem~\ref{th:egap} now identifies the exact optimal performance, as well as provides a more efficient way (see Remark~\ref{R:1}) to evaluate this performance.%, compared to \BS{the state of the art} (cf.~\cite[Theorem 1]{jin_toc_19}). 
\item Theorem~\ref{th:UBLB} offers a new tighter upper bound on $\overline{T}^*(\gamma)$ (see Figure~\ref{fig:soa}) and the only known lower bound on $\overline{T}^*(\gamma)$.
\item Finally Theorem~\ref{th:gap} completes our understanding of the scaling laws of the random association setting. For example, for the case where $K=\Theta\left(\Lambda\right)$, prior to our work, $G(\gamma)$ was known to scale at most as $\Theta\left(\sqrt{\Lambda}\right)$, whereas now we know that this deterioration scales exactly as $\Theta\left(\frac{\log \Lambda}{\log \log \Lambda} \right)$. Please refer to Table~\ref{tab:soa} for a detailed comparison of the known upper bounds and our exact scaling results. 

% \begin{table}[t]
% \centering
% \begin{tabular}{|l|l|l|}
% \hline
%   & $\overline{T}^*(\gamma)$ in \cite{jin_toc_19} & $\overline{T}^*(\gamma)$ in our work  \\ \hline
% $K= \Theta\left(\Lambda\right)$ &  $O\left( \sqrt{\Lambda}\right)$  & $ \Theta\left(\frac{\log \Lambda}{\log \log \Lambda} \right)$         \\\hline
% \begin{tabular}[c]{@{}l@{}}$K= \Theta\left(\Lambda^{a}\right)$ for\\ $1<a<2$ and \\ $K = \Omega\left(\Lambda\log \Lambda \right)$\end{tabular}&  $O\!\left( \Lambda^{a/2}\right)$  & \begin{tabular}[c]{@{}l@{}}$ \Theta\!\left(  T_{min}\right) =\Theta\left(\frac{K}{\Lambda}\right)$ \\ $= \Theta\!\left(\Lambda^{a-1} \right)$\end{tabular}       \\ \hline
% $K = \Omega\left(\Lambda^2\right)$ & $O\left( \frac{K}{\Lambda}\right)$            & $ \Theta\!\left(  T_{min}\right)= \Theta\left(  \frac{K}{\Lambda}\right)$ \\ \hline
% \end{tabular}
% \caption{SoA comparison of scaling laws.}
% \label{tab:soa}
% \end{table}
\begin{table}[t]
\centering
\begin{tabular}{|l|l|l|}
\hline
   & $\overline{T}^*(\gamma)$ in \cite{jin_toc_19} & $\overline{T}^*(\gamma)$ in our work  \\ \hline
$K= \Theta\left(\Lambda\right)$ &  $O\left( \sqrt{\Lambda}\right)$  & $ \Theta\left(\frac{\log \Lambda}{\log \log \Lambda} \right)$         \\\hline
$K= \Theta\left(\Lambda^{a}\right)$ for $1<a<2$ and $K = \Omega\left(\Lambda\log \Lambda \right)$&  $O\!\left( \Lambda^{a/2}\right)$  & $ \Theta\!\left(  T_{min}\right) =\Theta\left(\frac{K}{\Lambda}\right)$  $= \Theta\!\left(\Lambda^{a-1} \right)$       \\ \hline
$K = \Omega\left(\Lambda^2\right)$ & $O\left( \frac{K}{\Lambda}\right)$            & $ \Theta\!\left(  T_{min}\right)= \Theta\left(  \frac{K}{\Lambda}\right)$ \\ \hline
\end{tabular}
\caption{SoA comparison of scaling laws.}
\label{tab:soa}
\end{table}
\end{itemize}

\section{Cache Load Balancing in Heterogeneous Networks}\label{Sec:3}
In the previous section, we explored the performance of coded caching when each user is associated, at random and with equal probability, to one of $\Lambda$ caches. Our aim now is to reduce the detrimental impact of the user-to-cache association's randomness on the delivery time, by using load balancing methods that introduce a certain element of choice in this association, and thus allow for \emph{better} profile vectors. Such choice can exist naturally in different scenarios, like for example in the wireless cache-aided heterogeneous network setting, where each user can be within the communication range of more than one cache helper node. 
 
We define a generic load balancing method $\phi$ to be a function that maps the set of users $\left[K\right]$ into a cache population vector $\mathbf{V} = \phi([K])$ as a result of the load balancing choice. Similarly as in \eqref{eq:DefMetric}, the optimal delay, given a certain load balancing policy $\phi$, is defined as
\begin{equation}
\overline{T}_{\phi}^*(\gamma)  = \min_{\mathcal{X}} E_{\mathbf{V}}\left[ \max_{\mathbf{d}} T\left(\phi([K]),\mathbf{d},\mathcal{X}\right)\right]. \label{eq:DefMetric1.1} 
\end{equation}
The above definition is the same as the one in \eqref{eq:DefMetric}, with the only difference that the random variable representing the cache population vector $\mathbf{V}$ is now following a different probability distribution that depends on the load balancing method $\phi$. Employing the optimal scheme $\mathcal{X}$ from Theorem~\ref{th:egap}, the average delivery time takes the form (cf. equation \eqref{eq:tavg}) 
\begin{equation}\label{eq:Tstarloadbalancing}
\overline{T}_{\phi}(\gamma)  = \sum_{\lambda=1}^{\Lambda-t} E[ l_{\lambda} ]   \frac{ {\Lambda-\lambda \choose t} }{{\Lambda \choose t}},
\end{equation}
where $[l_1,l_2,\dots,l_{\Lambda}]=sort(\phi([K]))$.
It is important to point out that the choice of the load balancing method can be in general limited by some practical constraints, such as geographical constraints and operational constraints\footnote{Removal of all these constraints naturally brings us back to the ideal user-to-cache association where each cache is associated to an equal number of users.}.  We will focus on analyzing the above, for two load balancing methods which will prove to allow for unbounded gains.

\subsection{Randomized load balancing with multiple choices}\label{Sec:31}
In the original scenario, for any given user, one cache is randomly picked to assist this user. Now we consider a load balancing method $\phi_r$ which, for any given user, picks $h\geq 2$ candidate caches at random, and then associates each such user with the least loaded cache among these $h$ caches. This static method is referred to as \emph{randomized load balancing with multiple choices}~\cite{Mitzenmacher2001}, and is considered in a variety of settings (see for example~\cite{petra_06}).
The performance of this method is presented in the following result, for the limiting case of large $\Lambda$. 
\begin{theorem}\label{th:gapd}
In the $K$-user, $\Lambda$-cell heterogeneous network with normalized cache size $\gamma$, where each user benefits from the least loaded cache among $h$ randomly chosen caches, the limiting optimal delay converges to 
%\begin{align} \label{eq:lawtavgd}
%\overline{T}_{\phi_r}^{*}(\gamma) =  \Theta \left( T_{min} \frac {\Lambda \log \log \Lambda}{K  \log h} \right). 
%\end{align}
\begin{align} \label{eq:lawtavgd}
\!\!\overline{T}_{\phi_r}^{*}(\gamma)\!\! =\!
\begin{cases} 
\!\Theta \left( T_{min} \frac {\Lambda \log \log \Lambda}{K  \log h} \right) \text{ \emph{if} }\! K\!\! =\!o\left(\!\frac {\Lambda \log \log \Lambda}{\log h}\!\! \right) \\
\!\Theta \left(T_{min}\right) \hspace{1.5cm} \text{ \emph{if} }\! K\!\! =\!\Omega\left(\!\frac {\Lambda \log \log \Lambda}{\log h}\!\! \right).
\end{cases}
\end{align}
\end{theorem}
\begin{proof}
The achievability part of the theorem is deferred to Appendix~\ref{AP:gapd}. After noticing that the definition of the optimal delay in \eqref{eq:DefMetric1.1} is equal to \eqref{eq:DefMetric}, optimality is proven the same way as for the optimality of Theorem~\ref{th:egap} by following the same steps as in equations \eqref{eq:optim1}-\eqref{eq:TLopt}. Following those steps requires (cf. \cite{parrinello_it20}) that $P(\mathbf{V})$ remains fixed for any $\mathbf{V}$ such that $sort(\mathbf{V})=\mathbf{L}$; which is true also for the considered load balancing method $\phi_r$ because the method is not biased to any specific cache, i.e. $\phi_r$ assigns each user to one of the available caches only based on the load of the caches and independently from the cache identity. Therefore, the proof follows the same steps as for the case where there is no load balancing.
\end{proof} 
The above theorem naturally implies that the performance deterioration, due to random association, scales as 
%\begin{align} \label{eq:egapd}
%G_r(\gamma)=  \Theta \left( \frac {\Lambda \log \log \Lambda}{K  \log h} \right), 
%\end{align}
\begin{align} \label{eq:egapd}
G_r(\gamma)=  
\begin{cases} 
\!\Theta \left( \frac {\Lambda \log \log \Lambda}{K  \log h} \right) \text{ if } K =\!o\left(\frac {\Lambda \log \log \Lambda}{\log h} \right) \\
\!\Theta \left(1\right) \hspace{1.3cm} \text{ if } K =\!\Omega\left(\frac {\Lambda \log \log \Lambda}{\log h} \right),
\end{cases}
\end{align}
as well as implies the following corollary.
\begin{corollary}~\label{c:gapd}
In the $K$-user, $\Lambda$-cell heterogeneous network with random-selection load balancing, the performance deterioration due to random association, scales as $\Theta \left( \frac {\log \log \Lambda}{\log h} \right)$ when $K=\Theta\left(\Lambda\right)$, and then as $K$ increases, this deterioration gradually reduces, and ceases to scale when $K = \Omega\left( \frac {\Lambda \log \log \Lambda}{\log h} \right)$. 
\end{corollary}
\begin{proof}
The proof is direct from~\eqref{eq:egapd}.
\end{proof}
We can see that the above method can dramatically ameliorate the random association effect, where (for example when $K$ is in the same order as $\Lambda$) even a small choice among $h=2$ caches, can tilt the scaling of $G(\gamma)$, from the original $\Theta \left( \frac {\log \Lambda}{\log \log \Lambda} \right)$ to a much slower  $\Theta \left( \log \log \Lambda \right)$.

\subsection{Load balancing via proximity-based cache selection}\label{Sec:32}
The aforementioned randomized load balancing method, despite its substantial impact, may not apply when the choice is limited by geographical proximity. To capture this limitation, we consider the load balancing approach $\phi_p$ where the set of $\Lambda$ caches is divided into $\Lambda / h$ disjoint groups 
$\left[\mathbf{X}_1, \mathbf{X}_2, \dots, \mathbf{X}_{\Lambda / h} \right]$ of $h$ caches each\footnote{In this method, our focus is in the asymptotic setting, thus we do not need to assume that $h$ divides $\Lambda$.}. Once a user is associated at random, with uniform probability, to one of these groups, then we choose to associate this user to the least loaded cache from that group.

The performance of this method is presented in the following result, for the limiting case of large $\Lambda$. 
\begin{theorem}\label{th:gapg}
In the $K$-user, $\Lambda$-cell heterogeneous network with normalized cache size $\gamma$, where each user benefits from the least loaded cache among $h$ neighboring caches, then the limiting optimal delay converges to 
\begin{align} \label{eq:lawtavgg}
\overline{T}^{*}_{\phi_p}(\gamma)\!\! =\!
\begin{cases} 
\!\Theta\! \left(\frac{T_{min}\Lambda\log\frac{\Lambda}{h}}{hK\log\frac{\Lambda \log\frac{\Lambda}{h} }{hK}}\!\! \right)\! \text{ \emph{if} } K \!\in\!\left[ \frac{\Lambda}{h\polylog(\frac{\Lambda}{h})},  o\left(\!\frac{\Lambda}{h} \!\log\frac{\Lambda}{h}\! \right)\right] \\
\!\Theta \left(T_{min}\right) \hspace{1.15cm} \text{ \emph{if} } K =\!\Omega\left(\frac{\Lambda}{h} \log\frac{\Lambda}{h} \right).
\end{cases}
\end{align}
\end{theorem}
\begin{proof} 
The achievability proof is deferred to Appendix~\ref{AP:gapg}, while the optimality part of the theorem follows the same argument as the proof of Theorem~\ref{th:gapd}.
\end{proof} 
The above implies a performance deterioration of
\begin{align}\label{eq:egapg}
\!\!\!G_p(\gamma)\! =\!
\begin{cases}\! \Theta\! \left(\frac{\frac{\Lambda}{hK}\log\frac{\Lambda}{h}}{\log\frac{\Lambda \log\frac{\Lambda}{h} }{hK}} \right) \text{ if } K \in\left[ \frac{\Lambda / h}{\polylog(\frac{\Lambda}{h})},  o\left(\!\frac{\Lambda}{h} \!\log\frac{\Lambda}{h}\! \right)\right] \\
\!\Theta \left( 1\right) \hspace{1.35cm}\text{ if } K =\Omega\left(\frac{\Lambda}{h} \log\frac{\Lambda}{h} \right),
\end{cases}
\end{align}
which in turn implies the following.
\begin{corollary}~\label{c:gapg}
In the $K$-user, $\Lambda$-cell heterogeneous network with proximity-bounded load balancing, the performance deterioration due to random association scales as $\Theta \left(\frac{\log(\Lambda / h)}{ \log \log(\Lambda / h)} \right)$ when $K=\Theta\left(\frac{\Lambda}{h}\right)$, and as $K$ increases, this deterioration gradually reduces, and ceases to scale when $K= \Omega\left(\frac{\Lambda}{h}\log \frac{\Lambda}{h}\right)$. 
\end{corollary}
\begin{proof}
The proof is straightforward from Theorem~\ref{th:gapg}.
\end{proof}
We can see that proximity-bounded load balancing significantly ameliorate the random association effect, where now deterioration ceases to scale when $K= \Omega\left(\frac{\Lambda}{h}\log \frac{\Lambda}{h}\right)$ compared to the original $K= \Omega\left(\Lambda\log \Lambda\right)$.

\section{Numerical Validation}\label{Sec:4}
We proceed to numerically validate our results, using two basic numerical evaluation approaches. The first is the \emph{sampling-based numerical} (SBN) approximation method, where we generate a sufficiently large set $\mathcal{L}_1$ of randomly generated profile vectors $\mathbf{L}$, and approximate $E_{\mathbf{L}}[T(\mathbf{L})]$ as
\begin{align}\label{eq:atavgs} 
E_{\mathbf{L}}[T(\mathbf{L})] \approx \frac{1}{|\mathcal{L}_1|}\sum_{\mathbf{L} \in \mathcal{L}_1}  T(\mathbf{L}), 
 \end{align}
where we recall that $T(\mathbf{L})$ is defined in~\eqref{eq:TL}. The corresponding approximate performance deterioration is then evaluated by dividing the above by $T_{min}$.
\begin{figure}[t]
\centering
\includegraphics[ width=.99\linewidth]{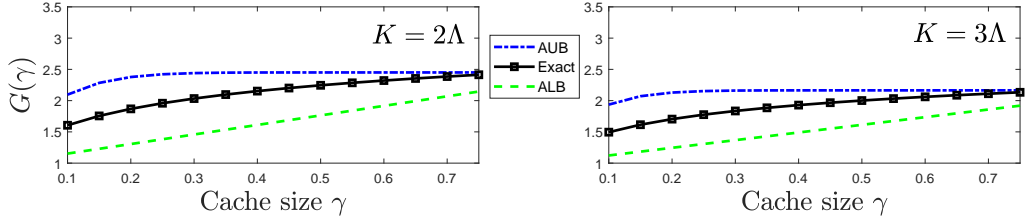} 
\caption{Analytical upper bound (AUB) from \eqref{eq:UB} vs. analytical lower bound (ALB) from \eqref{eq:LB} vs. exact $G(\gamma)$ from~\eqref{eq:gap} ($\Lambda=20$).}
\label{fig:ex_vs_ab}
\end{figure}
\begin{figure}[t]
\centering
\includegraphics[ width=.99\linewidth]{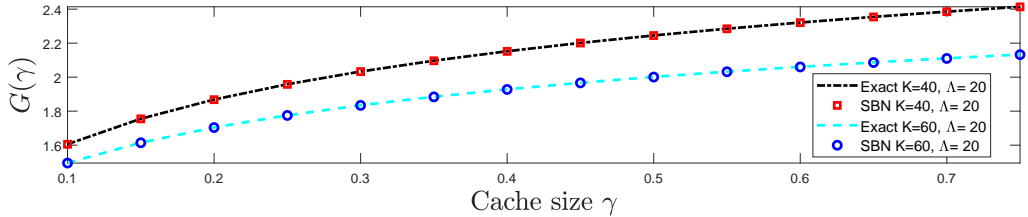} 
\caption{Exact $G(\gamma)$ from~\eqref{eq:gap} vs. sampling-based numerical (SBN) approximation from~\eqref{eq:atavgs} ($|\mathcal{L}_1|=10000$).}
\label{fig:ex_vs_sbn}
\end{figure} 
\begin{figure}[t]
\centering
\includegraphics[ width=\linewidth]{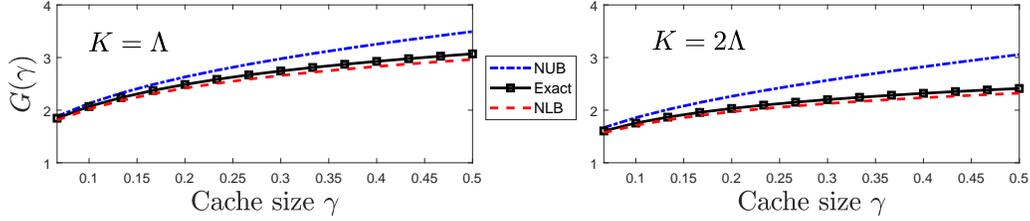} 
\caption{Threshold-based numerical upper bound (NUB) from~\eqref{eq:Nutavg} vs. threshold-based numerical lower bound (NLB) from~\eqref{eq:Nltavg} vs. exact $G(\gamma)$ from~\eqref{eq:gap} ($\Lambda=30$ and $\rho=0.95$).}
\label{fig:ex_vs_nb}
\end{figure}
The second is a \emph{threshold-based numerical} method, whose first step is to generate a set $\mathcal{L}_2 \subseteq \mathcal{L}$ of profile vectors $\mathbf{L}$ such that $\sum_{\mathbf{L} \in \mathcal{L}_2 } P(\mathbf{L}) \approx \rho$, for some chosen threshold value $\rho \in[0,1]$. Recall that the closed form expression for $P(\mathbf{L})$ is given in equation~\eqref{eq:pl}. Subsequently, with this subset $\mathcal{L}_2$ at hand, we simply have the numerical lower bound (NLB) 
\begin{align} \label{eq:Nltavg}
E_{\mathbf{L}}[T(\mathbf{L})] \geq    \sum_{\mathbf{L} \in \mathcal{L}_2 } P(\mathbf{L})  T(\mathbf{L}) +  \left(1-\rho \right) T_{min},
 \end{align}
by considering the best-case delay for each $\mathbf{L} \in \mathcal{L}/\mathcal{L}_2$, and similarly have the numerical upper bound (NUB) 
\begin{align} \label{eq:Nutavg}
E_{\mathbf{L}}[T(\mathbf{L})] \leq   \sum_{\mathbf{L} \in \mathcal{L}_2 } P(\mathbf{L})  T(\mathbf{L}) +  \left(1-\rho \right) K (1-\gamma),
 \end{align}
by considering the worst possible delay $K(1-\gamma)$ for every $\mathbf{L} \in \mathcal{L}/\mathcal{L}_2$. 
The bounding of $G(\gamma)$ is direct by dividing the above with $T_{min}$.  

Naturally the larger the threshold $\rho$, the tighter the bounds, the higher the computational cost. The additive gap between the bounds on $G(\gamma)$, takes the form $\left(1-\rho \right) \left(\frac{K (1-\gamma)}{T_{min}}-1 \right) \approx \left(1-\rho \right)t $, revealing the benefit of increasing $\rho$. 

First, Figures~\ref{fig:ex_vs_ab}-\ref{fig:ex_vs_nb} include comparisons that involve the \emph{exact} $G(\gamma)$ from~\eqref{eq:gap}, and thus --- due to the computational cost --- the number of caches remains at a modest $\Lambda=20$ (and a relatively larger $\Lambda=30$ for Figure~\ref{fig:ex_vs_nb}). In particular, Figure~\ref{fig:ex_vs_ab} compares the exact $G(\gamma)$ with the analytical bounds in~\eqref{eq:UB} and~\eqref{eq:LB}, where it is clear that both AUB and ALB yield sensible bounds, and AUB becomes much tighter as $\gamma$ increases. Figure~\ref{fig:ex_vs_sbn} compares the exact $G(\gamma)$ with the sampling-based numerical (SBN) approximation in~\eqref{eq:atavgs} (for $|\mathcal{L}_1|=10000$), where it is evident that the SBN approximation is consistent with the exact performance. Finally, Figure~\ref{fig:ex_vs_nb} compares the exact $G(\gamma)$ (for $\Lambda = 30$) with the threshold-based numerical bounds that are based on~\eqref{eq:Nltavg} and~\eqref{eq:Nutavg}, using $\rho=0.95$. Interestingly, the threshold-based NLB turns out to be very tight in the entire range of $\gamma$, whereas the NUB tends to move away from the exact performance as $\gamma$ increases.

\begin{figure}[t]
\centering
\includegraphics[ width=1\linewidth]{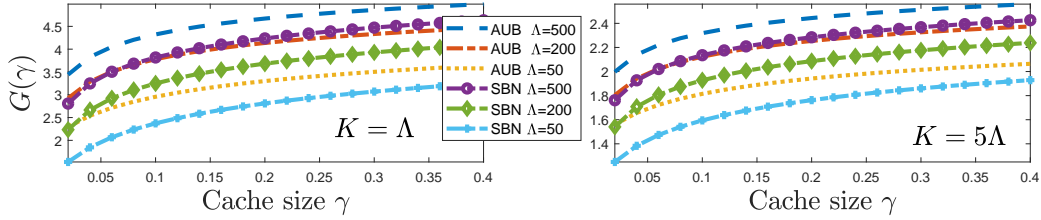} 
\caption{Analytical upper bound (AUB) from~\eqref{eq:UB} vs. sampling-based numerical (SBN) approximation from~\eqref{eq:atavgs} ($|\mathcal{L}_1|=10000$).}
\label{fig:aub_vs_sbn}
\end{figure}

Subsequently, for much larger dimensionalities, Figure~\ref{fig:aub_vs_sbn} compares the AUB from~\eqref{eq:UB} with the SBN approximation from~\eqref{eq:atavgs} for $|\mathcal{L}_1| = 10000$. Similarly, Figure~\ref{fig:alb_vs_sbn} compares the ALB from~\eqref{eq:LB} with the SBN approximation for $|\mathcal{L}_1| = 10000$. The figures highlight the extent to which the ratio $\frac{K}{\Lambda}$ affects the performance deterioration. 

Finally, Figure~\ref{fig:loadblnc} uses a suitably modified analytical upper bound to explore the effect of $h$ when applying proximity-bounded load balancing. We know from \eqref{eq:ubel} that the expected number of users in the most populous cache group (i.e, $E[l^h_1]$), when each user can be associated to any cache group with equal probability $\frac{h}{\Lambda}$ is bounded as
\begin{align}
 E[l^h_1]  \leq  K - \sum_{j=0}^{K-1} \max\left(1-\frac{\Lambda}{h}(1-P^{h}_j), 0\right),
\end{align}
where $P^{h}_j=\sum_{i=0}^{j} {K \choose i} \left(\frac{h}{\Lambda}\right)^i \left(1-\frac{h}{\Lambda}\right)^{K-i} $. Also, from \eqref{eq:bel1h}, the expected number of users in the most populous cache (i.e., $E[l_1]$) under proximity-bounded load balancing is bounded as $E[l_1] < \frac{E[l^h_1]}{h} +1$.
%\begin{align}
% E[l_1] < \frac{E[l^h_1]}{h} +1. 
 % E[l_1] < \frac{E[l^h_1]}{h} +1 = 1+ \frac{K}{h} - \frac{1}{h}\sum_{j=0}^{K-1} \max\left(1-\frac{\Lambda}{h}(1-P^{h}_j), 0\right). 
%\end{align}
Thus using \eqref{eq:TUB}, the analytical upper bound on the $\overline{T}^{*}_{\phi_p}(\gamma)$ is given by
\begin{align}\label{eq:UBP}
\overline{T}^{*}_{\phi_p}(\gamma) & \leq\frac{\Lambda-t}{1+t}  E[ l_{1} ] <\frac{\Lambda-t}{1+t}\left(\frac{E[ l^h_{1}]}{h} +1\right) \nonumber \\ 
&=\frac{\Lambda-t}{1+t}\left(1+ \frac{K}{h} - \frac{1}{h}\sum_{j=0}^{K-1} \max\left(1-\frac{\Lambda}{h}(1-P^{h}_j), 0\right)\right).
 \end{align} 
 From Figure~\ref{fig:loadblnc}, we can see that, as expected, the performance deterioration decreases as $h$ increases.
 \begin{figure}[t]
\centering
\includegraphics[ width=1\linewidth]{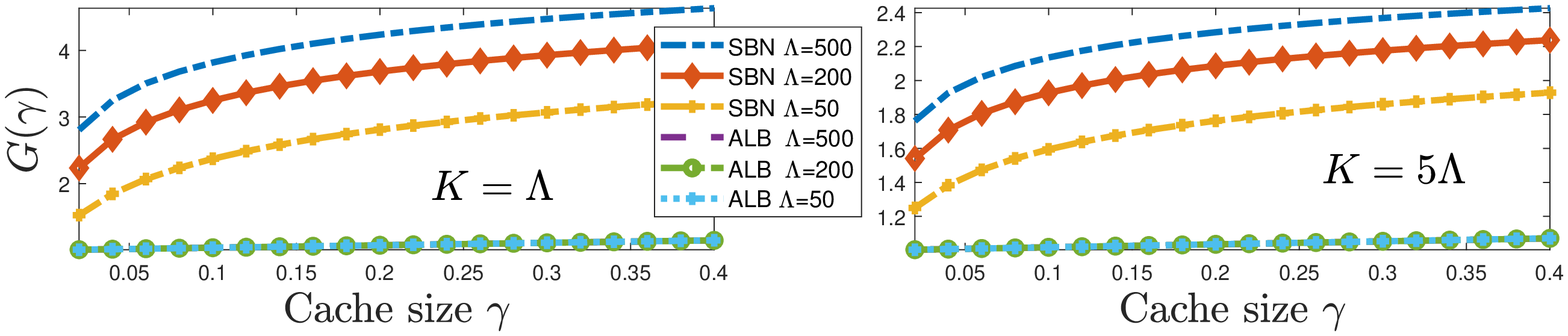} 
\caption{Analytical lower bound (ALB) from~\eqref{eq:LB} vs. sampling-based numerical (SBN) approximation from~\eqref{eq:atavgs} ($|\mathcal{L}_1|=10000$).}
\label{fig:alb_vs_sbn}
\end{figure}
\begin{figure}[t]
\centering
\includegraphics[width=1\linewidth]{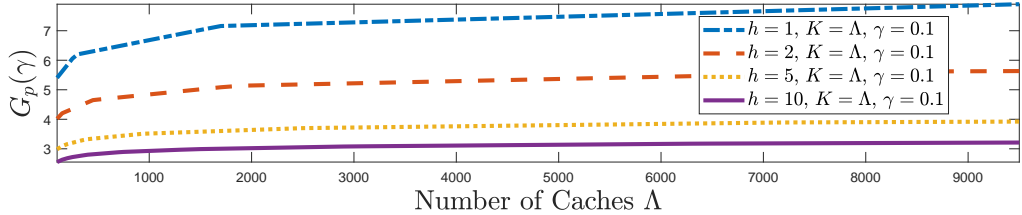} 
\caption{Analytical upper bound (AUB) from~\eqref{eq:UBP} without ($h=1$) and with ($h>1$) proximity-bounded load balancing.}
\label{fig:loadblnc}
\end{figure}

\section{Extension to the case of non-uniform cache population intensities}\label{Sec:NonUniform}
In this section, we extend our study to the scenario where cache population intensities (i.e, probability that a user can appear in the coverage area of any particular cache-enabled cell) are following a non-uniform distribution\footnote{All the results presented in this section are optimal for the case when the cache population intensities is not known during the cache placement phase.}. For any cache $\lambda \in [\Lambda]$, let $p_{\lambda}$ be the probability that a user can appear in the coverage area of $\lambda$th cache-enabled cell such that $\mathbf{p} =\left[p_1, p_2, \cdots, p_{\Lambda}\right]$, where $\sum_{\lambda \in [\Lambda]}p_{\lambda}=1$, denotes the cache population intensities vector.
\subsection{Analytical Bounds}
Considering the uncoded cache placement scheme in~\cite{man_it14}, and the delivery scheme in~\cite{parrinello_it20}, the following theorem bounds the average delay $\overline{T}(\gamma)$, when cache population intensities are following a non-uniform distribution.
%The following theorem bounds the average delay $\overline{T}(\gamma)$.
\begin{theorem}\label{th:UBLBnn}
In the $K$-user, $\Lambda$-cache setting with normalized cache size $\gamma$ and a random user-to-cache association with cache population intensities $\mathbf{p}$, the average delay $\overline{T}(\gamma)$ is bounded by 
\begin{align}\label{eq:UBtavgnn}
\overline{T}(\gamma) \leq K    \frac{\Lambda- t }{1+t}- \sum_{\lambda=1}^{\Lambda-t}    \frac{ {\Lambda-\lambda \choose t} }{{\Lambda \choose t}}\sum_{j=0}^{K-1}  \max \left(0,1- \frac{\Lambda- F(j)}{\lambda} \right),
\end{align}
and 
  \begin{align} \label{eq:LBtavgnn}
\overline{T}(\gamma)\geq \frac{\Lambda-t}{1+t} \left( \frac{K t \max(\mathbf{p}) }{(\Lambda-1) } +\frac{K}{\Lambda}  \frac{  (\Lambda-t-1)}{(\Lambda-1)} \right),
 \end{align}
\noindent where 
\begin{align}\label{eq:scdftnn}
F(j) = \sum_{k=1}^{\Lambda} \sum_{i=0}^{j} {K \choose i} \left(p_k\right)^i \left(1-p_k\right)^{K-i}.
 \end{align}
\end{theorem}
\begin{proof} 
The proof is deferred to Appendix \ref{AP:UBLBnn}.
\end{proof}
Directly from Theorem~\ref{th:UBLBnn} and equation~\eqref{eq:tbc0}, we can conclude that for $\frac{K}{\Lambda}\in \mathbb{Z}^{+}$, the performance deterioration $G(\gamma)$ as compared to the deterministic uniform case, is bounded as
\begin{align}\label{eq:UBnn}
G(\gamma) \leq \Lambda- \frac{1+t}{K-K\gamma} \sum_{\lambda=1}^{\Lambda-t}    \frac{ {\Lambda-\lambda \choose t} }{{\Lambda \choose t}}\sum_{j=0}^{K-1}  \max \left(0,1- \frac{\Lambda- F(j)}{\lambda} \right),
\end{align}
and 
 \begin{align} \label{eq:LBnn}
G(\gamma) \geq   \frac{\Lambda t \max(\mathbf{p}) }{(\Lambda-1) } + \frac{  (\Lambda-t-1)}{(\Lambda-1)} ,
 \end{align}
\noindent where $F(j)$ is given in Theorem~\ref{th:UBLBnn}.
It is fast to numerically evaluate the analytical bound proposed in Theorem~\ref{th:UBLBnn} for any given distribution of cache population intensities $\mathbf{p}$. However, in order to gain some simple and insightful form of the performance in the presence of non-uniform cache population intensities, we proceed with the asymptotic analysis of the $\overline{T}(\gamma)$ under the assumption that cache population intensities $\mathbf{p}$ follows the Zipf distribution\footnote{There are several studies that propose different user distribution models (i.e., distribution of cache population intensities) for wireless networks~\cite{li_16, george_16, zhou_15}. We use the Zipf distribution as it nicely covers a wide range of non-uniform patterns by only tuning one parameter.}. For the Zipf distribution, cache population intensities $\mathbf{p}$ are given by\footnote{Without loss of generality, we assume a descending order between cache population intensities of the $\Lambda$ caches.}
\begin{equation}\label{eq:zipf}
p_{\lambda} = \frac{\lambda^{-\alpha}}{H_{\alpha}(\Lambda)}, \hspace{0.5cm} \forall \ \ \lambda \in [\Lambda],
\end{equation}
\noindent where $\alpha>0$ is the Zipf exponent, and $H_{\alpha}(\Lambda) = \sum_{i=1}^{\Lambda} i^{-\alpha}$ is a normalization constant formed as the generalized harmonic number.
%Without loss of generality, we assume a descending order between request  cache population intensities of the $\Lambda$ caches.
%Riemann zeta function.
 %[ref] Distribution of number of users per cell in a poisson wireless network with shadowing
%[ref] Heterogeneous cellular network user distribution model
%C. Li, A. Yongacoglu and C. D'Amours, "Heterogeneous cellular network user distribution model," 2016 8th IEEE Latin-American Conference on Communications (LATINCOM), Medellin, 2016, pp. 1-6, doi: 10.1109/LATINCOM.2016.7811561.
% [ref] Distribution of the Number of Users per Base Station in Cellular Networks
%G. George, A. Lozano and M. Haenggi, "Distribution of the Number of Users per Base Station in Cellular Networks," in IEEE Wireless Communications Letters, vol. 8, no. 2, pp. 520-523, April 2019, doi: 10.1109/LWC.2018.2878579.
%[ref] On the Spatial Distribution of Base Stations and Its Relation to the Traffic Density in Cellular Networks
%S. Zhou, D. Lee, B. Leng, X. Zhou, H. Zhang and Z. Niu, "On the Spatial Distribution of Base Stations and Its Relation to the Traffic Density in Cellular Networks," in IEEE Access, vol. 3, pp. 998-1010, 2015, doi: 10.1109/ACCESS.2015.2452576.
\subsection{Scaling Laws}
The following theorem provides the asymptotic analysis of the $\overline{T}(\gamma)$, in the limit of large $\Lambda$.
\begin{theorem}\label{th:gapnn}
In the $K$-user, $\Lambda$-caches setting with normalized cache size $\gamma$ and random user-to-cache association with cache population intensities $\mathbf{p}$ following the Zipf distribution with the Zipf exponent $\alpha$, the delay scales as
\begin{align} \label{eq:lawtavgnn}
\overline{T}(\gamma)= \begin{cases}
   \Theta\left( T_{min} \Lambda    \right) & \alpha > 1 \\
          O\left( T_{min}  \sqrt{\frac{\Lambda^2}{K}+ \frac{\Lambda^2}{(\log \Lambda)^2}} \right)  
   \ \ \ \text{and} \ \ \    \Omega\left(T_{min} \frac{\Lambda}{\log \Lambda}\right)            & \alpha =1  \\
       O\left(T_{min}   \sqrt{ \frac{\Lambda^2}{K}+\Lambda^{2\alpha}}\right)\ \ \ \text{and} \ \ \   \Omega\left(T_{min} \Lambda^{\alpha} \right)  & 0.5<\alpha <1  
                \\
        O\left( T_{min}  \sqrt{\frac{\Lambda^2}{K} + \Lambda\log\Lambda}\right)  \ \ \  \text{and} \ \ \   \Omega\left(T_{min} \sqrt{\Lambda}  \right) & \alpha =0.5          \\
		O\left(T_{min} \left( \sqrt{\Lambda+ \frac{\Lambda^2}{K} }\right)\right) \ \ \  \text{and} \ \ \   \Omega\left(T_{min} \Lambda^{\alpha}   \right)  & \alpha <0.5. 	
\end{cases}
\end{align}
% \begin{align} \label{eq:lawtavgnn}
% \overline{T}(\gamma)= \begin{cases}
%   \Theta\left( T_{min} \Lambda    \right) & \alpha > 1 \\
%           O\left( T_{min}  \sqrt{\frac{\Lambda^2}{K}+ \frac{\Lambda^2}{(\log \Lambda)^2}} \right)              & \alpha =1 \\
%   \ \ \ \text{and} \ \ \    \Omega\left(T_{min} \frac{\Lambda}{\log \Lambda}\right) \\
%       % \Theta\left( T_{min} \frac{\Lambda}{\log \Lambda} \right)              & \alpha =1 \\
%       %O\left(T_{min}   \sqrt{ \frac{\Lambda^2}{K}+\Lambda^{2\alpha}}\right) \text{and} \   \Omega\left(T_{min} \Lambda^{\alpha} \right)          & 0.5<\alpha <1 \\
%       O\left(T_{min}   \sqrt{ \frac{\Lambda^2}{K}+\Lambda^{2\alpha}}\right) & 0.5<\alpha <1  \\
%       \ \ \ \text{and} \ \ \   \Omega\left(T_{min} \Lambda^{\alpha} \right)          \\
%         %O\left( T_{min}  \sqrt{\frac{\Lambda^2}{K} + \Lambda\log\Lambda}\right)  \text{and} \    \Omega\left(T_{min} \sqrt{\Lambda}  \right)         & \alpha =0.5 \\
%         O\left( T_{min}  \sqrt{\frac{\Lambda^2}{K} + \Lambda\log\Lambda}\right)  & \alpha =0.5 \\ \ \ \  \text{and} \ \ \   \Omega\left(T_{min} \sqrt{\Lambda}  \right)         \\
% 	%O\left(T_{min} \left(1  + \frac{\Lambda}{\sqrt{K}}\right)\right) \text{and} \  \Omega\left(T_{min} \Lambda^{\alpha}   \right)	     & \alpha <0.5.
% 		O\left(T_{min} \left( \sqrt{\Lambda+ \frac{\Lambda^2}{K} }\right)\right)  & \alpha <0.5. \\ \ \ \  \text{and} \ \ \   \Omega\left(T_{min} \Lambda^{\alpha}   \right)	
% \end{cases}
% \end{align}
\end{theorem}
\begin{proof} 
The proof is deferred to Appendix \ref{AP:gapnn}.
\end{proof}
Directly from the above, we now know that when the cache population intensities $\mathbf{p}$ follows the Zipf distribution, the performance deterioration due to user-to-cache association randomness, scales as
\begin{align} \label{eq:egapnn}
\!\!\!G(\gamma)\!= \begin{cases}
   \Theta\left( \Lambda    \right) & \alpha > 1 \\
    O\left(\!\sqrt{\frac{\Lambda^2}{K}\!+\! \frac{\Lambda^2}{(\log \Lambda)^2}} \right)   \text{and} \  \Omega\left(\frac{\Lambda}{\log \Lambda}\right)     & \alpha =1 \\
       % \Theta\left(\frac{\Lambda}{\log \Lambda} \right)              & \alpha =1 \\
           O\left(   \sqrt{ \frac{\Lambda^2}{K}+\Lambda^{2\alpha}}\right) \text{and} \  \Omega\left( \Lambda^{\alpha} \right)   & 0.5<\alpha <1  \\
             O\left(   \sqrt{\frac{\Lambda^2}{K} + \Lambda\log\Lambda}\right) \text{and} \   \Omega\left( \sqrt{\Lambda}  \right)  & \alpha =0.5           \\
	O\left(  \sqrt{\Lambda+ \frac{\Lambda^2}{K} }\right) \text{and} \  \Omega\left( \Lambda^{\alpha}   \right)  & \alpha <0.5.	    
\end{cases}
\end{align}

In identifying the scaling laws of the problem, Theorem~\ref{th:gapnn} nicely captures the following points:
\begin{itemize}
    \item It describes to what extent the performance deterioration increases with $\alpha$ (i.e., the skewness in cache population intensities).
    \item It shows that, in some cases there is no global caching gain, e.g., the performance deterioration scales as $\Theta\left(\Lambda\right)$ for $\alpha>1$.
    \item It reveals that unlike the case of uniform cache population intensities -- where the deterioration can be avoided as long as $K=\Omega\left(\Lambda\log \Lambda\right)$ -- the existence of skewness in cache population intensities can lead to an unbounded deterioration irrespective of the relation between $K$ and $\Lambda$.%, e.g., when $\alpha\geq 1$. 
    \item It highlights the importance of incorporating the knowledge of the cache population intensities vector while designing the placement and delivery scheme. As pointed out earlier, being unaware of the severeness of this non-uniformity may lead to the vanishing of the coding gain, and the system may eventually need to confine itself to the local caching gain.
\end{itemize} 

\subsection{Randomized load balancing with multiple choices under non-uniform cache population intensities}\label{Sec:51}
 We consider a load balancing method $\phi_n$ which, for any given user, picks $h\geq 2$ candidate caches at random based on the cache population intensities $\mathbf{p}$ following the Zipf distribution, and then associates each such user to the least loaded cache among these $h$ caches. The performance of this method is presented in the following result, for the limiting case of large $\Lambda$. 
\begin{theorem}\label{th:gapdnn}
In the $K$-user, $\Lambda$-cell heterogeneous network with normalized cache size $\gamma$, where each user benefits from the least loaded cache among $h$ randomly chosen caches based on the cache population intensities $\mathbf{p}$ with the Zipf exponent $\alpha$, the limiting delay converges to 
\begin{align} \label{eq:lawtavgdnn}
\!\!\overline{T}_{\phi_n}(\gamma)\!\! =\!
\begin{cases} 
\!O \left( T_{min} \frac {\Lambda \log \log \Lambda}{K} \right) \text{ \emph{if} }\! K\!\! =\!o\left(\Lambda \log \log \Lambda \right) \\
\!O \left(T_{min}\right) \hspace{1.5cm} \text{ \emph{if} }\! K\!\! =\!\Omega\left(\Lambda \log \log \Lambda\right),
\end{cases}
\end{align}
when $h=\Theta\left(\log \Lambda\right)$.
\end{theorem}
\begin{proof}
The proof is deferred to Appendix~\ref{AP:gapdnn}.
\end{proof}
The above theorem naturally implies that, if $h$ is in the same order as $\log \Lambda$ then the performance deterioration, due to random association, scales as 
%The above theorem naturally implies that for the given Zipf exponent $\alpha$, if $h$ scales according to \eqref{eq:hnn} then the performance deterioration, due to random association, scales as 
\begin{align} \label{eq:egapdnn}
G_n(\gamma)=\begin{cases} 
\!O \left(  \frac {\Lambda \log \log \Lambda}{K} \right) \text{ \emph{if} }\! K\!\! =\!o\left(\Lambda \log \log \Lambda \right) \\
\!O \left(1\right) \hspace{1.5cm} \text{ \emph{if} }\! K\!\! =\!\Omega\left(\Lambda \log \log \Lambda\right).
\end{cases}  
\end{align}
We can see that load balancing can dramatically ameliorate the random association effect. For example, when $\alpha>1$, picking any $\log \Lambda $ candidate caches is sufficient to tilt the scaling of $G(\gamma)$ from $\Theta \left( \Lambda \right)$ to a much slower $O \left( \log \log \Lambda \right)$ and $O \left( 1 \right)$, when $K=\Theta(\Lambda)$ and $K =\Omega\left(\Lambda \log \log \Lambda\right)$ respectively. As long as $h$ is in the same order as $\log \Lambda$, significant improvements can be achieved irrespective of the level of skewness of the cache population intensities. In conclusion, even for the non-uniform cache population intensities, load balancing can still be impactful. However, for non-uniform cache population intensities setting $h=2$ may not bring significant gains which were observed for the case of uniform cache population intensities case (cf. Corollary~\ref{c:gapd} ) as now $h$ must be in the order of $\log \Lambda$.

\begin{figure}[t]
\centering
\includegraphics[ width=1\linewidth]{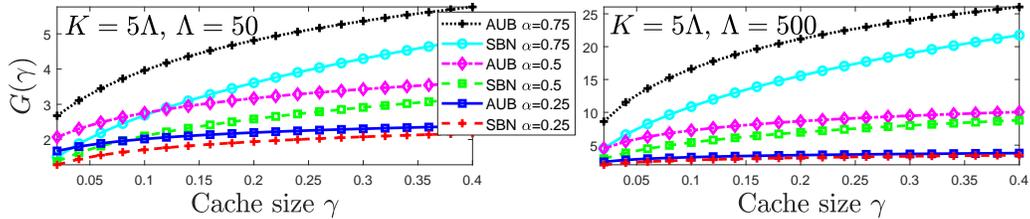} 
\caption{Analytical upper bound (AUB) from~\eqref{eq:UBnn} vs. sampling-based numerical (SBN) approximation from~\eqref{eq:atavgs} ($|\mathcal{L}_1|=10000$).}
\label{fig:nn_aub_vs_sbn}
\end{figure}
\begin{figure}[t]
\centering
\includegraphics[ width=1\linewidth]{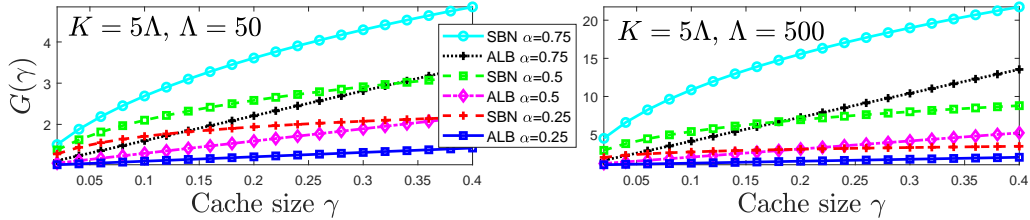} 
\caption{Analytical lower bound (ALB) from~\eqref{eq:LBnn} vs. sampling-based numerical (SBN) approximation from~\eqref{eq:atavgs} ($|\mathcal{L}_1|=10000$).}
\label{fig:nn_alb_vs_sbn}
\end{figure}
\subsection{Numerical validation for the non-uniform cache population intensities}\label{Sec:34}
 We now numerically validate our results for the case of non-uniform cache population intensities. For the numerical analysis, we assume that cache population intensities $\mathbf{p}$ follows the Zipf distribution. Figure~\ref{fig:nn_aub_vs_sbn} compares the AUB from~\eqref{eq:UBnn} with the SBN approximation from~\eqref{eq:atavgs} for $|\mathcal{L}_1| = 10000$. Note that $\mathcal{L}_1$ is generated based on cache population intensities $\mathbf{p}$. Similarly, Figure~\ref{fig:nn_alb_vs_sbn} compares the ALB from~\eqref{eq:LBnn} with the SBN approximation for $|\mathcal{L}_1| = 10000$. The figures highlight the extent to which the Zipf exponent $\alpha$ (i.e., the skewness in cache population intensities) affects the performance deterioration.

\section{Conclusions}\label{Sec:6}
In this work we identified the exact optimal performance of coded caching with random user-to-cache association when users can appear in the coverage area of any particular cache-enabled cell with equal probability. In our opinion, the random association problem has direct practical ramifications, as it captures promising scenarios (such as the heterogeneous network scenario) as well as operational realities (namely, the subpacketization constraint). The problem becomes even more pertinent as we now know that its effect can in fact scale indefinitely.

Key to our effort to identify the effect of association randomness, has been the need to provide expressions that can either be evaluated in a numerically tractable way, or that can be rigorously approximated in order to yield clear insight. The first part was achieved by deriving exact expressions as well as new analytical bounds that can be evaluated directly, while the second part was achieved by studying the asymptotics of the problem which yielded simple performance expressions and direct operational guidelines. This same approach also allowed us to clearly capture the effect and importance of basic load balancing techniques that are used to mitigate the detrimental effect coming from the aforementioned randomness in user-to-cache association.

Finally, we extended our analysis for the case where cache population intensities are following a non-uniform distribution. We provided analytical bounds and studied the asymptotics of the problem. Also, we show that load balancing can help mitigating the effect of randomness even when cache population intensities are non-uniform.

\appendix
%\section{Appendix}
%\renewcommand\theequation{\Alph{section}.\arabic{equation}}
%\setcounter{equation}{0}

\subsection{Proof of Theorem \ref{th:egap}} \label{AP:egap}
We first note that the probability $P(\mathbf{L})$ of observing a specific profile vector $\mathbf{L}\in \mathcal{L}$ is simply the cumulative probability over all $\mathbf{V}$ for which $sort(\mathbf{V}) = \mathbf{L}$. This probability takes the form
\begin{align}\label{eq:pl} 
P(\mathbf{L}) =\overbrace{ \frac{1}{\Lambda^{K}}  \ \times  \frac{K!}{\prod _{i=1}^{\Lambda}l_i!} }^\text{term 1} \    \   \  \times \overbrace{     \     \frac{\Lambda!}{\prod _{j=1}^{\left|B_L\right|}b_j!}      }^\text{term 2}.
\end{align}

To see this, we analyze the different terms of the above equation. The first term in~\eqref{eq:pl} accounts for the fact that there are $\Lambda^{K}$ different user-to-cache associations, i.e., there are $\Lambda^{K}$ different ways that $K$ users can be allocated to the $\Lambda$ different caches. It also accounts for the fact that each user can be associated to any one particular cache, with equal probability $\frac{1}{\Lambda}$. 
The second term in~\eqref{eq:pl} indicates the number of all user-to-cache associations that leads\footnote{Recall that different user-to-cache associations can lead to the same cache population vector $\mathbf{V}$. For example, when $K= \Lambda=3$, the following $6$ user-to-cache associations, $\left[1,2,3\right]$, $\left[1,3,2\right]$,  $\left[2,1,3\right]$, $\left[2,3,1\right]$,  $\left[3,2,1\right]$, and  $\left[3,1,2\right]$ --- each describing which user is associated to which cache --- in fact all correspond to the same $\mathbf{V}= \left[1, 1, 1 \right]$, because always each cache is associated to one user.} to a specific $\mathbf{V}$ for which $sort(\mathbf{V})=\mathbf{L}$, for some fixed $\mathbf{L}$. 
Consequently term 1 in~\eqref{eq:pl} is simply $P(\mathbf{V})$, which naturally remains fixed for any $\mathbf{V}$ for which $sort(\mathbf{V})= \mathbf{L}$, and which originates from the well-known probability mass function of the multinomial distribution. Consequently this implies that $P(\mathbf{L})= |\{ \mathbf{V} \!:\! sort(\mathbf{V})= \mathbf{L} \}|\times P(\mathbf{V})$.  Finally, term $2$ describes the number of all possible cache population vectors $\mathbf{V}$ for which $sort(\mathbf{V})$ is equal to some fixed $\mathbf{L}$.

We now proceed to insert~\eqref{eq:pl} into~\eqref{eq:tavg}, which yields the average delay 
\begin{align}
E_{\mathbf{L}}[T(\mathbf{L})]&=  \sum_{\lambda=1}^{\Lambda-t} \sum_{\mathbf{L} \in \mathcal{L} }P(\mathbf{L}) l_{\lambda}   \frac{ {\Lambda-\lambda \choose t} }{{\Lambda \choose t}}  =  \sum_{\lambda=1}^{\Lambda-t} \sum_{\mathbf{L} \in \mathcal{L}} \frac{l_{\lambda} K! \Lambda! }{\Lambda^{K} \prod _{i=1}^{\Lambda}l_i! \prod _{j=1}^{\left|\mathbf{B}_{\mathbf{L}}\right|}b_j!  }  \frac{ {\Lambda-\lambda \choose t} }{{\Lambda \choose t}} \nonumber \\
&= \sum_{\lambda=1}^{\Lambda-t} \sum_{\mathbf{L} \in \mathcal{L}} \frac{K! \  t! \  (\Lambda-t)! \ l_{\lambda}  {\Lambda-\lambda \choose t}}{ \Lambda^{K}\prod _{i=1}^{\Lambda}l_i! \prod _{j=1}^{\left|\mathbf{B}_{\mathbf{L}}\right|}b_j! },
\end{align}
which concludes the achievability part of the proof for the expression in Theorem~\ref{th:egap}. 

Optimality of the aforementioned expression can be proved by means of the lower bound developed in \cite{parrinello_it20}. We notice that the optimal delay $\overline{T}^*(\gamma)$ can be lower bounded as
\begin{align}\nonumber
\overline{T}^*(\gamma)&=\min_{\mathcal{X}} E_{\mathbf{L}}\left[
                     E_{\mathbf{V}_{\mathbf{L}}}\left[  \max_{\mathbf{d}}T(\mathbf{\mathbf{V},\mathbf{d}},\mathcal{X})\right]\right] \geq \min_{\mathcal{X}} E_{\mathbf{L}}\left[
                  \max_{\mathbf{d}} E_{\mathbf{V}_{\mathbf{L}}}\left[  T(\mathbf{\mathbf{V},\mathbf{d}},\mathcal{X})\right]\right]\\
&\geq E_{\mathbf{L}}\left[\min_{\mathcal{X}}
                  \max_{\mathbf{d}} E_{\mathbf{V}_{\mathbf{L}}}\left[  T(\mathbf{\mathbf{V},\mathbf{d}},\mathcal{X})\right]\right]\geq E_{\mathbf{L}}\bigg[\underbrace{\min_{\mathcal{X}}
                  E_{\mathbf{d}\in \mathcal{D}_{wc}} E_{\mathbf{V}_{\mathbf{L}}}\left[  T(\mathbf{\mathbf{V},\mathbf{d}},\mathcal{X})\right]}_{T^*(\mathbf{L})}\bigg],    \label{eq:optim1}              
\end{align}
where $\mathcal{D}_{wc}$ denoted the set of demand vectors with distinct users' file-requests. Next, exploiting the fact that  $P(\mathbf{V})$ is the same for any $\mathbf{V}$ for which $sort(\mathbf{V})=\mathbf{L}$, we notice that 
\[T^*(\mathbf{L}) \triangleq \min_{\mathcal{X}}
                  E_{\mathbf{d}\in \mathcal{D}_{wc}} E_{\mathbf{V}_{\mathbf{L}}}\left[  T(\mathbf{\mathbf{V},\mathbf{d}},\mathcal{X})\right] \] is lower bounded by equation $(53)$ in~\cite{parrinello_it20}, which then proves that $T^*(\mathbf{L})$ is bounded as
\begin{align}\label{eq:TLopt}
 T^*(\mathbf{L}) \geq    \sum_{\lambda=1}^{\Lambda-t} l_{\lambda} \frac{{\Lambda-\lambda \choose t} }{{\Lambda \choose t}}.
 \end{align}
This concludes the proof for the optimality of the delivery time in Theorem \ref{th:egap}.

\subsection{Proof of Theorem \ref{th:UBLB}} \label{AP:UBLB}
We start our proof by deriving the expected number of users in the $\lambda$-th most populous cache (i.e., $E[l_{\lambda} ]$), which is given by
\begin{align}\label{eq:el}
E[l_{\lambda}] &= \sum_{j=0}^{K-1} P[l_{\lambda} > j]  = \sum_{j=0}^{K-1} \left(1- P[l_{\lambda} \leq j] \right) = K- \sum_{j=0}^{K-1}  P[l_{\lambda} \leq j].
\end{align}
where $P[l_{\lambda}\leq j]$ is the probability that $\lambda$-th most populous cache is associated to no more than $j$ requesting users. From \cite[Proposition 2]{caraux_92}, we have 
\begin{align}\label{eq:aprub}
   P[l_{\lambda} \leq j] \geq \max\left(1-\frac{\Lambda}{\lambda}(1-P_j), 0\right),
 \end{align}
 where $P_j$ is the probability that a cache is associated to no more than $j$ requesting users. Recalling that each user can be assigned to any particular cache with equal probability, we can conclude that $P_j$ is given as
 \begin{align}\label{eq:cdf}
P_j  =\sum_{i=0}^{j} {K \choose i} \left(\frac{1}{\Lambda}\right)^i \left(1-\frac{1}{\Lambda}\right)^{K-i}.
 \end{align}
$E[l_{\lambda}]$ is upper bounded by 
\begin{align}\label{eq:ubel}
E[l_{\lambda}] &\leq K- \sum_{j=0}^{K-1} \max\left(1-\frac{\Lambda}{\lambda}(1-P_j), 0\right).   
\end{align}
Consequently the upper bound of $\overline{T}^*(\gamma)$ is given as 
\begin{align}
\overline{T}^*(\gamma)& =\sum_{\lambda=1}^{\Lambda-t} E[ l_{\lambda} ]  \frac{ {\Lambda-\lambda \choose t} }{{\Lambda \choose t}}  \leq  \sum_{\lambda=1}^{\Lambda-t} \frac{ {\Lambda-\lambda \choose t} }{{\Lambda \choose t}} \left(K-\sum_{j=0}^{K-1}  \max\left(1-\frac{\Lambda}{\lambda}(1-P_j), 0\right)   \right)\nonumber\\
 &=  \sum_{\lambda=1}^{\Lambda-t} \frac{ {\Lambda-\lambda \choose t} }{{\Lambda \choose t}} K - \sum_{\lambda=1}^{\Lambda-t} \frac{ {\Lambda-\lambda \choose t} }{{\Lambda \choose t}}  \sum_{j=0}^{K-1} \max\left(1-\frac{\Lambda}{\lambda}(1-P_j), 0\right) \nonumber \\
   &\stackrel{(a)}{=} \frac{ {\Lambda\choose t+1} }{{\Lambda \choose t}} K - \sum_{\lambda=1}^{\Lambda-t} \frac{ {\Lambda-\lambda \choose t} }{{\Lambda \choose t}}  \sum_{j=0}^{K-1}  \max\left(1-\frac{\Lambda}{\lambda}(1-P_j), 0\right) \nonumber \\
   &=  K\frac{ \Lambda-t }{ t+1} - \sum_{\lambda=1}^{\Lambda-t} \frac{ {\Lambda-\lambda \choose t} }{{\Lambda \choose t}}  \sum_{j=0}^{K-1} \max\left(1-\frac{\Lambda}{\lambda}(1-P_j), 0\right), 
\end{align} %http://www.cs.columbia.edu/~cs4205/files/CM4.pdf
% \begin{align}
% \overline{T}^*(\gamma)& =\sum_{\lambda=1}^{\Lambda-t} E[ l_{\lambda} ]  \frac{ {\Lambda-\lambda \choose t} }{{\Lambda \choose t}}  \leq  \sum_{\lambda=1}^{\Lambda-t} \frac{ {\Lambda-\lambda \choose t} }{{\Lambda \choose t}} \left(\!\!K\!-\!\! \sum_{j=0}^{K-1}\!  \max\left(1\!-\!\frac{\Lambda}{\lambda}(1-P_j), 0\right)   \right)\nonumber\\
%  &=  \sum_{\lambda=1}^{\Lambda-t} \frac{ {\Lambda-\lambda \choose t} }{{\Lambda \choose t}} K\! -\! \!\sum_{\lambda=1}^{\Lambda-t} \!\frac{ {\Lambda-\lambda \choose t} }{{\Lambda \choose t}} \! \sum_{j=0}^{K-1} \!\! \max\left(\!\!1\!-\!\frac{\Lambda}{\lambda}(1\!-\!P_j), 0\!\right) \nonumber \\
%   &\stackrel{(a)}{=} \frac{ {\Lambda\choose t+1} }{{\Lambda \choose t}} K \!-\! \sum_{\lambda=1}^{\Lambda-t} \frac{ {\Lambda-\lambda \choose t} }{{\Lambda \choose t}}  \sum_{j=0}^{K-1}  \max\left(\!1-\frac{\Lambda}{\lambda}(1\!-\!P_j), 0\!\right) \nonumber \\
%   &=  K\frac{ \Lambda\!-\!t }{ t\!+\!1}\! - \!\!\sum_{\lambda=1}^{\Lambda-t} \frac{ {\Lambda-\lambda \choose t} }{{\Lambda \choose t}} \!\! \sum_{j=0}^{K\!-\!1}\! \! \max\left(\!\!1\!-\frac{\Lambda}{\lambda}(1\!-\!P_j), 0\!\right), \!\!\!\!
% \end{align}
where in step (a), we used the column-sum property of Pascal's triangle, which is $\sum^{n}_{k=0} {k\choose t}= {n+1\choose t+1}$. This concludes the proof of the upper bound in \eqref{eq:UBtavg}. 

Next, we prove the lower bound in \eqref{eq:LBtavg}. Crucial to this proof is the exploitation of the fact that $\sum_{\lambda=1}^{\Lambda} E[ l_{\lambda} ]= K$ and of the fact that both $E[ l_{\lambda}]$ and ${\Lambda-\lambda \choose t}$ in \eqref{eq:tavg} are non-increasing with $\lambda$. We first see that
\begin{align} 
\overline{T}^*(\gamma)&=\!\! \sum_{\lambda=1}^{\Lambda-t} E[ l_{\lambda} ]    \frac{ {\Lambda-\lambda \choose t} }{{\Lambda \choose t}}   \geq   \frac{  E[ l_{1} ]   {\Lambda-1 \choose t} \!+ \! \sum_{\lambda=2}^{\Lambda-t} B  {\Lambda-\lambda \choose t} }{{\Lambda \choose t}},  
 \end{align}
where $B=\frac{K-E[ l_{1}]}{\Lambda-1}$. This can be simplified as 
\begin{align} \label{eq:TLB}
\overline{T}^*(\gamma) &\geq   \frac{  E[ l_{1} ]   {\Lambda-1 \choose t}+  \sum_{\lambda=2}^{\Lambda-t} B  {\Lambda-\lambda \choose t} }{{\Lambda \choose t}} =   E[ l_{1} ]  \frac{   {\Lambda-1 \choose t}} {{\Lambda \choose t}}  + B  \frac{  {\Lambda-1 \choose t+1}}{{\Lambda \choose t}} \nonumber \\
&=  E[ l_{1} ]  \frac{   \Lambda-t} {\Lambda}  + B  \frac{  (\Lambda-t) (\Lambda-t-1)}{(1+t)\Lambda}  =  (\Lambda-t) \left( \frac{ E[ l_{1} ]}{\Lambda}  + B  \frac{  \Lambda-t-1}{(1+t)\Lambda} \right)\nonumber \\
&=  (\Lambda-t) \left( \frac{ E[ l_{1} ]}{\Lambda}  +\frac{K-E[ l_{1}]}{\Lambda-1}  \frac{  \Lambda-t-1}{(1+t)\Lambda} \right)=  \frac{\Lambda-t}{1+t} \left( \frac{E[ l_{1} ] t  }{ \Lambda-1} +\frac{K}{\Lambda}  \frac{  \Lambda-t-1}{\Lambda-1} \right).
%&=  \frac{K(1-\gamma)}{1+t}  \left( \frac{\Lambda E[ l_{1} ] t  }{K (\Lambda-1)} +  \frac{  (\Lambda-t-1)}{(\Lambda-1)} \right).
 \end{align}
To conclude the proof, we need to derive $E[l_1]$. It is straightforward that $l_1\geq \left\lceil \frac{K}{\Lambda} \right\rceil$, thus for $j= \left[0, 1, 2, \cdots, \left\lceil \frac{K}{\Lambda} \right\rceil-1 \right]$ we have 
\begin{align}
P[l_1 \leq j] = 0,
 \end{align}
 and for $j= \left[ \left \lceil \frac{K}{\Lambda} \right\rceil,  \left\lceil \frac{K}{\Lambda} \right\rceil+1, \cdots, K \right]$, from \cite[Proposition 1]{caraux_92}  we have 
 \begin{align}\label{eq:aprlb}
  P[l_{1} \leq j] &\leq \min(P_j,1) = P_j,\end{align}
 where $P_j$ is defined in \eqref{eq:cdf}. Therefore, using \eqref{eq:el}, $E[l_{1}]$ is lower bounded as 
\begin{align}\label{eq:lbel1}
E[l_{1}] = K- \sum_{j=0}^{K-1}  P[l_1 \leq j]   \geq  K- \sum_{j=\left \lceil \frac{K}{\Lambda} \right\rceil}^{K-1}  P_j.
\end{align}
Finally, combining \eqref{eq:TLB} and \eqref{eq:lbel1}, we obtain
\begin{align} 
\overline{T}^*(\gamma)\geq \frac{\Lambda\!-\!t}{1+t}\!\!\left(\!\! \frac{t}{\Lambda\!-\!1}\!\left(\!\!K\!-\!\!\!\! \!\sum_{j=\left \lceil \frac{K}{\Lambda} \right\rceil}^{K-1} \!\! P_j\!\! \right) \!\!+\!\frac{K}{\Lambda}  \frac{  \Lambda\!-\!t\!-\!1}{\Lambda-1}\!\! \right),
 \end{align}
 which concludes the proof of Theorem~\ref{th:UBLB}.
\subsection{Proof of Theorem \ref{th:gap}} \label{AP:gap}
The fact that both $E[ l_{\lambda}]$ and ${\Lambda-\lambda \choose t}$ in \eqref{eq:tavg} are non-increasing with $\lambda$, we see that $\overline{T}^*(\gamma)$ is bounded by
\begin{align}\label{eq:TUB}
\!\!\overline{T}^*(\gamma) &=\! \sum_{\lambda=1}^{\Lambda-t} E[ l_{\lambda} ]  \frac{ {\Lambda-\lambda \choose t} }{{\Lambda \choose t}} \nonumber \leq \!  \frac{1 }{{\Lambda \choose t}} \sum_{\lambda=1}^{ \Lambda-t} E[ l_{1} ]   {\Lambda-\lambda \choose t} \stackrel{(a)}{=}\frac{E[ l_{1} ] }{{\Lambda \choose t}}  \sum_{\lambda=1}^{\Lambda-t}  {\Lambda-\lambda \choose t}   
 \nonumber  \\&=   E[ l_{1} ] \frac{  {\Lambda \choose t+1}}{{\Lambda \choose t}} =   E[ l_{1} ] \frac{\Lambda-t}{1+t} = \frac{K(1-\gamma)}{1+t} \frac{\Lambda E[ l_{1} ]}{K},
 \end{align}
where in step (a), we used the column-sum property of Pascal's triangle, which is $\sum^{n}_{k=0} {k\choose t}= {n+1\choose t+1}$. Thus from~\eqref{eq:TUB}, we get
\begin{align}\label{eq:UBa}
\overline{T}^{*}(\gamma) =  O \left(  \frac{K(1-\gamma)}{1+t} \frac{ \Lambda E[ l_{1} ]}{K} \right)
\end{align}
and from~\eqref{eq:TLB}, we have 
\begin{align}\label{eq:LBa}
\overline{T}^{*}(\gamma) = \Omega \left(  \frac{K(1-\gamma)}{1+t} \frac{ \Lambda E[ l_{1} ] \gamma }{K } \right).
\end{align}
As $\gamma$ is a constant, we can conclude that the expressions in~\eqref{eq:UBa} and \eqref{eq:LBa} asymptotically match, and thus
\begin{align}\label{eq:gapt0}
\overline{T}^{*}(\gamma)  =  \Theta\left( \frac{K(1-\gamma)}{1+t} \frac{ E[ l_{1} ] \Lambda}{K} \right).
\end{align}
Combining \eqref{eq:gapt0} and \eqref{eq:tbc}, we obtain 
\begin{align}\label{eq:gapt}
\overline{T}^{*}(\gamma) =  \Theta\left( T_{min} \frac{ E[ l_{1} ] \Lambda}{K} \right).
\end{align}
For the remaining part, which is to develop the asymptotics of $E[l_1]$, we proceed with the following lemma which is adopted and adapted here directly from the work of~\cite{martin_98} on the \emph{Balls into Bins problem}.
\begin{lemma}[\!\!\!\protect{\cite[Theorem 1]{martin_98} - adaptation}]\label{le:ballbin}
In a $\Lambda$-cell $K$-user setting where each user can be associated with equal probability to any of the caches, the tail of $l_1$ takes the form
\begin{align} \label{eq:acdfl1}
P[l_1 > k_{\beta}]= 
\begin{cases} 
 o\left( 1\right) \hspace{0.6cm}\text{ if } \ \   \beta  > 1 \\
 1- o\left( 1\right) \text{ if }     \ \ 0 < \beta < 1,
\end{cases}
\end{align} 
for 
\begin{align} \label{eq:ka}
\!\!k_{\beta}\!= \!
\begin{cases} 
\!\frac{\log\Lambda}{\log\frac{\Lambda \log\Lambda }{K}} \! \left(\! 1\!+ \!\beta\frac{\log\log\frac{\Lambda \log\Lambda }{K}}{\log\frac{\Lambda \log\Lambda }{K}}\! \right)\! \text{ \emph{if} } \! \frac{\Lambda}{\polylog(\Lambda)}\!\leq\! K \! =\!\! o\left(\Lambda \log\Lambda \right) \\
\Theta \left(\beta\log\Lambda\right)  \  \text{ \emph{if} }\!  K \! = \!\Theta \left( \Lambda \log(\Lambda)\right) \\
\frac{K}{\Lambda}\!+\!  \beta \sqrt{\frac{K \log(\Lambda)}{0.5\Lambda}}\  \text{ \emph{if} }\! \omega\left(\Lambda \log\Lambda \right) =  K \! \leq \Lambda \polylog(\Lambda)\!\\
\!\frac{K}{\Lambda}\!+\!  \sqrt{\frac{K \log(\Lambda)}{0.5\Lambda}\!\! \left(\! 1\!-\! \frac{\log\log\Lambda }{2\beta\log\Lambda}\right)}\!  \text{ \emph{if} }\!  K \! = \!\omega\left(\!\Lambda \left(\log\Lambda\right)^3\! \right). 
\end{cases}
\end{align}
\end{lemma}
\begin{proof}
The result comes directly from \cite[Theorem 1]{martin_98}.
\end{proof}
With Lemma~\ref{le:ballbin} at hand, we consider the case of $\beta > 1$, for which we get that 
\begin{align} \label{eq:ubl1}
E[l_1] &= \! \sum_{j=0}^{k_{\beta}-1}\!\! P[l_1 > j]\! +\!  P[l_1 > k_{\beta} ] +\sum_{j=k_{\beta}+1}^{K-1}\!\! P[l_1 > j] \stackrel{(a)}{\leq} k_{\beta} +  o(1) +\sum_{j=k_{\beta}+1}^{K-1} P[l_1 > j] \nonumber \\
&\stackrel{(b)}{\leq} k_{\beta} +  o(1) + (K-k_{\beta}-1)o(1)= k_{\beta} (1-o(1)) + K o(1)  = O\left( k_{\beta}\right), 
\end{align}
where in step (a), we use the fact that $P[l_1 > j]$ is at most $1$ for $j= \left[0, 1, \cdots k_{\beta}-1\right]$ and in step (b), we use the fact if $P[l_1 > k_{\beta}] =o(1)$ then $P[l_1 > j]$ is at most $o(1)$  for $j= \left[ k_{\beta}+1, \cdots K-1\right]$. Similarly, for $ 0 < \beta < 1$,  we have 
\begin{align}\label{eq:lbl1}
E[l_1] &= \sum_{j=0}^{k_{\beta}-1} P[l_1 > j] +  P[l_1 > k_{\beta} ] +\sum_{j=k_{\beta}+1}^{K-1} P[l_1 > j] \nonumber \\
&\stackrel{(a)}{\geq} k_{\beta} (1-o(1)) + 1- o(1) \geq k_{\beta} (1-o(1)) = \Omega\left( k_{\beta}\right), 
\end{align}
where in step (a), we use the fact that $\sum_{j=k_{\beta}+1}^{K-1} P[l_1 > j] \geq 0$ and if $P[l_1 > k_{\beta}] =1-o(1)$ then $P[l_1 > j]$ is at least $1-o(1)$ for $j= \left[0, 1, \cdots k_{\beta}-1\right]$. Combining~\eqref{eq:ka},~\eqref{eq:ubl1}, and~\eqref{eq:lbl1}, we have
\begin{align}
\!\!E[l_1]\! =\!
\begin{cases} 
\Theta \left(\frac{\log\Lambda}{\log\frac{\Lambda \log\Lambda }{K}} \right) \text{ if } \! \frac{\Lambda}{\polylog(\Lambda)}\leq K \! = o\left(\Lambda \log\Lambda \right) \\
\Theta \left(\log\Lambda\right) \  \text{ if }\!  K \! = \!\Theta \left( \Lambda \log(\Lambda)\right) \\
\Theta \left(\!\!\frac{K}{\Lambda}\!+ \!  \sqrt{\!\frac{K \log(\Lambda)}{\Lambda}}\!\right)\!  \text{ if }\! \omega\left(\Lambda \log\Lambda\! \right) \!= \! K \! \leq \!\Lambda \polylog(\Lambda)\!\\
\Theta \left(\!\frac{K}{\Lambda}+  \sqrt{\frac{K \log(\Lambda)}{\Lambda}}\right)  \text{ if }\!  K \! = \!\omega\left(\Lambda \left(\log\Lambda\right)^3 \right),
\end{cases}
\end{align}
\noindent which in turn implies that 
\begin{align}\label{eq:AL1}
\!\! \! E[l_1]\! =\!
\begin{cases} 
\!\Theta\! \left(\frac{\log\Lambda}{\log\frac{\Lambda \log\Lambda }{K}} \right)\! \text{ if } K \in[ \frac{\Lambda}{\polylog(\Lambda)}, o\left(\!\Lambda \!\log\Lambda\! \right)] \\
\!\Theta\! \left(\!\frac{K}{\Lambda}+  \sqrt{\frac{K \log(\Lambda)}{\Lambda}}\right) \text{ if }\! K =\Omega\left(\Lambda \log\Lambda \right).
\end{cases}
\end{align}
Combining~\eqref{eq:gapt} with~\eqref{eq:AL1}, allows us to directly conclude the proof of Theorem~\ref{th:gap}.

\subsection{Proof of Theorem \ref{th:gapd}} \label{AP:gapd}
Directly from the result in~\cite[Corollary 1.4]{petra_06} on the \emph{Balanced Allocations} problem, we can conclude that for $h>1$, the $E[l_1]$ asymptotically converges to
\begin{align}\label{eq:AL1d}
E[l_1]= \frac {\log \log \Lambda}{\log h}+\frac { K }{\Lambda} \pm \Theta (1).
\end{align}
\noindent Consequently combining~\eqref{eq:gapt} and~\eqref{eq:AL1d}, directly yields~\eqref{eq:lawtavgd} which concludes the proof of Theorem~\ref{th:gapd}.

\subsection{Proof of Theorem \ref{th:gapg}} \label{AP:gapg}
We start our proof by deriving the expected number of users in the most populous cache (i.e., $E[l_1]$). Recall that under the proximity-based load balancing technique, each user can be associated to any cache \emph{group} with equal probability $\frac{h}{\Lambda}$. Once a user is associated to a group, then this user will be associated to the least loaded cache from that group. Let $l^h_1$ be the number of users that are associated to the most populous group of caches, then the number of users in the most populous cache is given by $l_1=\left\lceil\frac{l^h_1}{h} \right\rceil$. Thus, we have 
\begin{align}
E[l_1] = \sum_{i=1}^{K} P[l^h_1=i]\left\lceil\frac{l^h_1}{h} \right\rceil, 
\end{align}
where $P[l^h_1=i]$ is the probability that $i$ users are associated to the most populous group of caches. Let $S_1 \subseteq [K]$ be the set of elements such that for each element $i\in S_1$, $\frac{i}{h}$ is integer. Then, we have
\begin{align}
E[l_1]& = \!\sum_{i \in S_1} P[l^h_1=i]\frac{l^h_1}{h} +\sum_{i \in S_1/[K]} P[l^h_1=i]\left\lceil\frac{l^h_1}{h} \right\rceil \nonumber \\
&= \!\sum_{i \in S_1} P[l^h_1=i]\frac{l^h_1}{h} +\sum_{i \in S_1/[K]} P[l^h_1=i]\frac{l^h_1}{h} +\sum_{i \in S_1/[K]} P[l^h_1=i] \left(\left\lceil\frac{l^h_1}{h} \right\rceil  \!-\!\frac{l^h_1}{h}\right) \nonumber \\
&= \!\sum_{i \in [K]}\! P[l^h_1\!=\!i]\frac{l^h_1}{h} \!+\!\!\!\!\!\sum_{i \in S_1/[K]}\!\!\!\! P[l^h_1\!=\!i] \left(\left\lceil\frac{l^h_1}{h} \right\rceil \! -\!\frac{l^h_1}{h}\right) = \frac{E[l^h_1]}{h} \!+\!\!\!\!\!\sum_{i \in S_1/[K]}\!\!\!\! P[l^h_1\!=\!i] \left(\left\lceil\frac{l^h_1}{h} \right\rceil \! -\!\frac{l^h_1}{h}\right).\!\!
\end{align}
It is straightforward to see that $  0 <\sum_{i \in S_1/[K]} P[l^h_1=i] \left(\left\lceil\frac{l^h_1}{h} \right\rceil  -\frac{l^h_1}{h}\right)< 1$, Therefore,  $E[l_1]$ is bounded as
\begin{align}\label{eq:bel1h}
\frac{E[l^h_1]}{h} < E[l_1] < \frac{E[l^h_1]}{h} +1,
\end{align}
and we can conclude that 
\begin{align}\label{eq:el1h}
E[l_1] = \Theta \left(\frac{E[l^h_1]}{h} \right). 
\end{align}
Evaluating $E[l^h_1]$ from \eqref{eq:AL1} by treating each group as a single cache, we conclude that 
\begin{align}\label{eq:AL1g}
E[l_1]\! =\!
\begin{cases} 
\!\Theta\! \left(\frac{\log\frac{\Lambda}{h}}{h\log\frac{\Lambda \log\frac{\Lambda}{h} }{hK}} \right)\! \text{ if } K \in \left[\frac{\Lambda / h}{\polylog(\frac{\Lambda}{h})},  o\left(\!\frac{\Lambda}{h} \!\log\frac{\Lambda}{h}\! \right)\right] \\
\!\Theta\! \left(\!\frac{K}{\Lambda}+  \sqrt{\frac{K \log(\frac{\Lambda}{h})}{h\Lambda}}\right) \text{ if }\! K =\Omega\left(\frac{\Lambda}{h} \log\frac{\Lambda}{h} \right).
\end{cases}
\end{align}
\noindent Finally combining \eqref{eq:gapt} and \eqref{eq:AL1g}, directly yields~\eqref{eq:lawtavgg}, and thus concludes the proof of Theorem~\ref{th:gapg}.

\subsection{Characterization of $T_{min}$} \label{AP:tmin}
From equation~\eqref{eq:TL}, we see the fact that $l_{\lambda}$ and ${\Lambda-\lambda \choose t}$ are non-increasing with $\lambda$, which implies that the profile vector $\mathbf{L}$, which minimizes the delay has components of the form
\begin{align}\label{eq:Ltbc}
l_{\lambda} =
\begin{cases} 
 \left\lfloor\frac{K}{\Lambda}\right\rfloor +1\!\!\!  \hspace{0.2cm} \text{ for }  \lambda \in \left[1, 2,  \dots, \hat{K}\right]   \\
 \left\lfloor\frac{K}{\Lambda}\right\rfloor  \hspace{0.6cm} \text{ for }  \lambda \in \left[\hat{K}\!+1, \hat{K}\!+2, \dots, \Lambda\right],
\end{cases}
\end{align} 
where $\hat{K}=K- \left\lfloor\frac{K}{\Lambda}\right\rfloor\Lambda$. Consequently, when $\hat{K} \geq \Lambda-t$, the corresponding best-case delay $T_{min}$ is given as
\begin{align}\label{eq:tminc2}
T_{min} =   \sum_{\lambda=1}^{\Lambda-t} \left(\left\lfloor\frac{K}{\Lambda}\right\rfloor+1\right) \frac{ {\Lambda-\lambda \choose t} }{{\Lambda \choose t}} =  \left( \left\lfloor\frac{K}{\Lambda}\right\rfloor +1 \right)\frac{\Lambda-t}{1+t},
\end{align}
while when $\hat{K} < \Lambda-t$, this is given as 
\begin{align}\label{eq:tminc1}
&T_{min} =    \sum_{\lambda=1}^{\hat{K}} \left(\left\lfloor\frac{K}{\Lambda}\right\rfloor+1\right) \frac{ {\Lambda-\lambda \choose t} }{{\Lambda \choose t}}+\sum_{\lambda=\hat{K}+1}^{\Lambda-t} \left\lfloor\frac{K}{\Lambda}\right\rfloor  \frac{ {\Lambda-\lambda \choose t} }{{\Lambda \choose t}} \nonumber\\
&=   \left\lfloor\frac{K}{\Lambda}\right\rfloor \sum_{\lambda=1}^{\Lambda-t}  \frac{ {\Lambda-\lambda \choose t} }{{\Lambda \choose t}} + \sum_{\lambda=1}^{\hat{K}}  \frac{ {\Lambda-\lambda \choose t} }{{\Lambda \choose t}}  =    \left( \left\lfloor\frac{K}{\Lambda}\right\rfloor +1 \right) \sum_{\lambda=1}^{\Lambda-t} \frac{ {\Lambda-\lambda \choose t} }{{\Lambda \choose t}}  -\sum_{\lambda=\hat{K}+1}^{ \Lambda-t }  \frac{ {\Lambda-\lambda \choose t} }{{\Lambda \choose t}}  \nonumber \\
& =   \left( \left\lfloor\frac{K}{\Lambda}\right\rfloor +1 \right) \frac{  {\Lambda \choose t+1}}{{\Lambda \choose t}}  - \frac{ {\Lambda-\hat{K} \choose t+1}}{{\Lambda \choose t}}  =   \frac{\Lambda-t}{1+t}  \left( \left\lfloor\frac{K}{\Lambda}\right\rfloor +1  - \frac{\prod_{i=t+1}^{\hat{K}+t} (\Lambda - i)}{ \prod_{j=0}^{\hat{K}-1} (\Lambda - j)} \right),
\end{align}
which reverts back to the well-known delay
\begin{align}\label{eq:tminc3}
T_{min} =     \frac{K}{\Lambda} \frac{\Lambda-t}{1+t},
\end{align}
when $\hat{K}=0$. This concludes the characterization of the best-case delay $T_{min}$.

\subsection{Proof of Theorem \ref{th:UBLBnn}} \label{AP:UBLBnn}
The proof is essentially the same as of Theorem~\ref{th:UBLB}'s proof with some minor changes.  Crucial to the proof of upper bound is the result from \cite[Proposition 3]{caraux_92}, which is given as 
  \begin{align}\label{eq:aprubnn}
P[l_{\lambda} \leq j] \geq  \max \left(0,1- \frac{\Lambda- \sum_{k=1}^{\Lambda}F_{v_{k}}(j) }{\lambda} \right),
\end{align}
 where $F_{v_{k}}(j)$ is the probability that $k$th cache is associated to no more than $j$ requesting users (i.e., $P[v_{k} \leq j]$ ). Recalling that a user can be associated to $k$th cache with probability $p_k$, we can conclude that $F_{v_{k}}(j)$ is given as
   \begin{align}\label{eq:cdfnn}
F_{v_{k}}(j) =\sum_{i=0}^{j} {K \choose i} \left(p_k\right)^i \left(1-p_k\right)^{K-i}.
 \end{align}
 We obtain the upper bound in \eqref{eq:UBtavgnn} by simply replacing \eqref{eq:aprub} with \eqref{eq:aprubnn} and following the same procedure as in the proof of Theorem~\ref{th:UBLB}. This concludes the proof of the upper bound in \eqref{eq:UBtavgnn}. Next, we prove the lower bound in \eqref{eq:LBtavgnn}. We know from \eqref{eq:TLB} that
\begin{align} \label{eq:TLBnn}
\overline{T}(\gamma) &\geq  \frac{\Lambda-t}{1+t} \left( \frac{E[ l_{1} ] t  }{ \Lambda-1} +\frac{K}{\Lambda}  \frac{  \Lambda-t-1}{\Lambda-1} \right).
 \end{align}
To conclude the proof, we need to derive $E[l_1]$. Let $m\in[\Lambda]$ is the cache with highest population intensity (i.e, $m = \underset{x \in [\Lambda]}{\mathrm{argmax}} \ p_x $), then $E[l_1]$ is lower bounded as
%To conclude the proof, we need to derive $E[l_1]$. Let $m\in[\Lambda]$ is the cache with highest population intensity (i.e, $m = \argmax \limits_{x \in [\Lambda]} p_x $), then $E[l_1]$ is lower bounded as
\begin{align}\label{eq:lbel1nn}
E[l_{1}] &= \sum_{\mathbf{V}\in \mathcal{V}}P(\mathbf{V})\max(\mathbf{V}) \geq \sum_{\mathbf{V}\in \mathcal{V}}P(\mathbf{V})v_m = E[v_m]= Kp_m.
\end{align}
Finally, we obtain
\begin{align} 
\overline{T}(\gamma)\geq \frac{\Lambda-t}{1+t}\left( \frac{Kp_m t}{\Lambda-1} +\frac{K}{\Lambda}  \frac{  \Lambda-t-1}{\Lambda-1} \right),
 \end{align}
 which concludes the proof of Theorem~\ref{th:UBLBnn}.
\subsection{Proof of Theorem \ref{th:gapnn}} \label{AP:gapnn}
We know from \eqref{eq:gapt} that $\overline{T}(\gamma) =  \Theta\left( T_{min} \frac{ E[ l_{1} ] \Lambda}{K} \right)$, thus, we first characterize the scaling laws of $E[l_{1}]$. Let $\mu_i$ and $\sigma^2_i$ be the expectation and variance of random variable $v_i$'s (i.e., the number of users associated to the $i$th cache) respectively, then from \cite[Proposition 1]{gascuel_92}, we have
\begin{align}
 E[l_{1}] \leq \bar{\mu}  + \sqrt{\frac{\Lambda-1}{\Lambda}\sum_{i=1}^{\Lambda} \left(\sigma_i^2+ (\mu_i-\bar{\mu})^2\right) },
\end{align}
where $\bar{\mu}= \frac{1}{\Lambda}\sum_{i=1}^{\Lambda}\mu_i$. Since the random variable $v_i$ follows the binomial distribution, we have $\mu_{i}= Kp_{i}$, $\sigma^2_{i}=K(1-p_{i})p_{i}$, and $\bar{\mu} = \frac{1}{\Lambda}\sum_{i=1}^{\Lambda} Kp_i = \frac{K}{\Lambda}$. Consequently the upper bound on the $E[l_{1}]$ is given as
\begin{align}\label{eq:el1nnub}
 E[l_{1}] &\leq \frac{K}{\Lambda} \! + \!\sqrt{\frac{\Lambda\!-\!1}{\Lambda}\sum_{i=1}^{\Lambda} \left(\!K(1\!-\!p_{i})p_{i}\!+\! \left(\!Kp_{i}\!-\!\frac{K}{\Lambda}\right)^2\right) }\! \nonumber\\
 &= \frac{K}{\Lambda} \! +\sqrt{\Lambda\!-\!1}\sqrt{\frac{K}{\Lambda}\sum_{i=1}^{\Lambda} \left(p_{i}\!-\!p^2_{i}\!+\! K\left(\!p^2_{i}\!-\!\frac{2p_{i}}{\Lambda}\!+\!\frac{1}{\Lambda^2} \!\right)\!\right) } \nonumber\\
 &= \frac{K}{\Lambda} \! +\!\sqrt{\Lambda\!-\!1}\sqrt{\frac{K}{\Lambda}\sum_{i=1}^{\Lambda} \left(  p_{i}\!-\!p^2_{i}\!+\! Kp^2_{i}-\frac{2Kp_{i}}{\Lambda}+\frac{K}{\Lambda^2} \right) } \nonumber\\
 &= \frac{K}{\Lambda} \! +\!\sqrt{\Lambda\!-\!1}\sqrt{\frac{K}{\Lambda}\sum_{i=1}^{\Lambda} \left(\!  p_{i}\!\left(\!1\!-\!\frac{2K}{\Lambda}\!\right)\!\! +\! p^2_{i}( K\!-\!1)\!+\!\frac{K}{\Lambda^2} \!\right) }\nonumber\\
% &=\frac{K}{\Lambda}  +\sqrt{\Lambda-1} \sqrt{\frac{K}{\Lambda} \left[ \sum_{i=1}^{\Lambda} p_{i}\left(1-\frac{2K}{\Lambda}\right) + \sum_{i=1}^{\Lambda}p^2_{i}( K-1)+\sum_{i=1}^{\Lambda}\frac{K}{\Lambda^2} \right] }\nonumber\\ 
 &=\frac{K}{\Lambda} \! +\!\sqrt{\Lambda\!-\!1} \sqrt{\frac{K}{\Lambda}\! \left(\! 1\!-\!\frac{2K}{\Lambda}\!+\!\frac{K}{\Lambda}  \!+\! \sum_{i=1}^{\Lambda}p^2_{i}( K\!-\!1)\!\right) }\nonumber\\
 &= \frac{K}{\Lambda} \! +\sqrt{\Lambda\!-\!1}\sqrt{\frac{K}{\Lambda} \left( 1-\frac{K}{\Lambda} + ( K-1) \sum_{i=1}^{\Lambda}p^2_{i}\right) }\nonumber\\
 &= \frac{K}{\Lambda} \! +\!\sqrt{\Lambda\!-\!1}\sqrt{ \!\frac{ K(\Lambda\!- \! K)}{\Lambda^2}\! +\! \frac{K(K\!-\!1)}{\Lambda} \! \sum_{i=1}^{\Lambda}p^2_{i} }\nonumber\\
 &= \frac{K}{\Lambda}\!  +\!\sqrt{\Lambda\!-\!1}\sqrt{ \frac{ K(\Lambda\!-\!  K) \!+\!K(K\!-\!1)\Lambda\sum_{i=1}^{\Lambda}p^2_{i}}{\Lambda^2}   }\nonumber\\
 %&=\frac{K}{\Lambda}  +\sqrt{\Lambda-1}\frac{1}{\Lambda}\sqrt{  K(\Lambda-  K) +K(K-1)\Lambda\sum_{i=1}^{\Lambda}p^2_{i}}\nonumber\\
 &=\frac{K}{\Lambda} \! +\!\sqrt{\!\Lambda\!-\!1}\frac{K}{\Lambda}\sqrt{  \frac{\Lambda\!-\!  K \!+\!(K\!-\!1)\Lambda\sum_{i=1}^{\Lambda}p^2_{i}}{K}}=\frac{K}{\Lambda}\left(\! \!  1\! +\!\sqrt{\!\Lambda\!-\!1\!}\sqrt{\! \frac{\Lambda}{K}\!- \! 1\! +\!\frac{K\!-\!1}{K}\Lambda \sum_{i=1}^{\Lambda}p^2_{i}\!} \right).
\end{align}    
 When the cache population intensities $\mathbf{p}$ follows the Zipf distribution, we have $\sum_{i=1}^{\Lambda}p^2_{i} = \sum_{i=1}^{\Lambda}\frac{i^{-2\alpha}}{\left(H_{\alpha}(\Lambda)\right)^2}=\frac{H_{2\alpha}(\Lambda)}{\left(H_{\alpha}(\Lambda)\right)^2}$, where the normalization constant $H_{\alpha}(\Lambda)$ known as the generalized harmonic number \cite{gitzenis_it13} scales as
 %Riemann zeta function scale as 
 %Asymptotic Laws for Joint Content Replication and Delivery in Wireless Networks
 %S. Gitzenis, G. Paschos, and L. Tassiulas, “Asymptotic laws for joint content replication and delivery in wireless networks,” IEEE Trans. Inf. Theory, vol. 59, no. 5, pp. 2760–2776, May 2013.
 %M. Mahdian and E.M. Yeh,, “Throughput and delay scaling of contentcentric ad hoc and heterogeneous wireless networks,” IEEE/ACM Trans. Netw., vol. 25, no. 5, pp. 3030–3043, Oct. 2017.
\begin{equation} \label{eq:H} 
H_{\alpha}(\Lambda) = \begin{cases}
    \Theta\left(1 \right) & \alpha > 1 \\
    \Theta\left(\log \Lambda \right)              & \alpha =1 \\
		\Theta\left( \Lambda^{1-\alpha} \right)  & \alpha <1.
\end{cases} 
\end{equation}
Similarly, the scaling of $H_{2\alpha}(\Lambda)$ is given as 
\begin{equation} \label{eq:H2} 
H_{2\alpha}(\Lambda) = \begin{cases}
    \Theta\left(1 \right) & \alpha > 0.5 \\
    \Theta\left(\log \Lambda \right)              & \alpha =0.5 \\
		\Theta\left( \Lambda^{1-2\alpha} \right)  & \alpha <0.5.
\end{cases} 
\end{equation}
Consequently, we have
\begin{align} \label{eq:H2r} 
\sum_{i=1}^{\Lambda}p^2_{i}= \begin{cases}
    \Theta\left(1 \right) & \alpha > 1 \\
        \Theta\left(\frac{1}{(\log \Lambda)^2} \right)              & \alpha =1 \\
        \Theta\left(\frac{1}{\Lambda^{2-2\alpha}} \right)              & 0.5<\alpha <1 \\
        \Theta\left(\frac{\log \Lambda }{\Lambda^{2-2\alpha}}\right)=  \Theta\left(\frac{\log \Lambda }{\Lambda}\right)             & \alpha =0.5 \\
		\Theta\left(\frac{\Lambda^{1-2\alpha}}{\Lambda^{2-2\alpha}} \right)=\Theta\left(\frac{1}{\Lambda} \right)  & \alpha <0.5.
\end{cases} 
\end{align}
From \eqref{eq:el1nnub} and \eqref{eq:H2r}, we obtain 
\begin{align} 
E[l_{1}]&= \begin{cases}
    O\left(\frac{K}{\Lambda}\left( 1 +\sqrt{\Lambda}\sqrt{ \frac{\Lambda}{K}\! +\!\Lambda } \right)\!\!\right)  & \!\!\!\!\alpha > 1 \\
        O\left(\frac{K}{\Lambda}\left(\! 1 \!+\!\sqrt{\Lambda}  \sqrt{ \frac{\Lambda}{K} \!+\! \frac{\Lambda}{(\log \Lambda)^2}}   \right)\! \!\right)          &\!\!\!\! \alpha =1 \\
        O\left(\frac{K}{\Lambda}\left(\! 1 \!+\!\sqrt{\Lambda}\sqrt{ \frac{\Lambda}{K} \!+\!\frac{\Lambda }{\Lambda^{2-2\alpha}}}      \right)\!\!\right)        & \!\!\!\!0.5\!<\!\alpha\! <\!1 \\
          O\left(\frac{K}{\Lambda}\left(\! 1 \!+\!\sqrt{\Lambda} \sqrt{ \frac{\Lambda}{K}\! +\!\frac{\Lambda  \log \Lambda }{\Lambda}}  \right)\!\! \right)         &\!\!\!\! \alpha =0.5 \\
	O\left(\frac{K}{\Lambda}\left( \!1 \!+\!\sqrt{\Lambda} \sqrt{ \frac{\Lambda}{K} \!+\!\frac{\Lambda }{\Lambda} }  \right) \!\!\right)  &\!\!\!\! \alpha <0.5
\end{cases} = \begin{cases}
      O\left(\frac{K}{\Lambda}\left( \Lambda \right)\right) &\!\!\!\! \alpha > 1 \\
         O\left(\frac{K}{\Lambda}\left( \!1 \!+\!\sqrt{\Lambda}  \sqrt{\frac{\Lambda}{K}\!+\! \frac{\Lambda}{(\log \Lambda)^2}}   \right) \!\!   \right)        &\!\!\!\! \alpha =1 \\
       O\left(\frac{K}{\Lambda}\left( \!1 \!+\!\sqrt{\Lambda}\sqrt{ \frac{\Lambda}{K}\!+\!\frac{1}{\Lambda^{1-2\alpha}}}      \right)   \!\!   \right)        &   \!\!\!\!0.5\!<\!\alpha\! <\!1 \\
         O\left(\frac{K}{\Lambda}\left(\! 1\! +\!\sqrt{\Lambda} \sqrt{ \frac{\Lambda}{K} \!+ \!\log \Lambda}  \right)  \!\! \right)        &\!\!\!\! \alpha =0.5 \\
	O\left(\frac{K}{\Lambda}\left( \!1\! +\!\sqrt{\Lambda} \sqrt{ \frac{\Lambda}{K} \!+1\!}  \right) \!\!\right) &\!\!\!\! \alpha <0.5
\end{cases} \nonumber\\
& = \begin{cases}
   O\left( K  \right) & \!\!\!\!\alpha > 1 \\
        O\left(\frac{K}{\Lambda}\left(1  + \sqrt{\frac{\Lambda^2}{K}+ \frac{\Lambda^2}{(\log \Lambda)^2}} \right)    \right)           & \!\!\!\!\alpha =1 \\
       O\left( \frac{K}{\Lambda}\left(1  + \sqrt{ \frac{\Lambda^2}{K}+\Lambda^{2\alpha}}\right)       \right)                & \!\!\!\!0.5\!<\!\alpha\! <\!1 \\
        O\left( \frac{K}{\Lambda}\left(1  + \sqrt{ \frac{\Lambda^2}{K} + \Lambda\log \Lambda}\right)    \right)           & \!\!\!\!\alpha =0.5 \\
	O\left(\frac{K}{\Lambda}\left(1  +  \sqrt{\Lambda+ \frac{\Lambda^2}{K}}\right)	    \right) &\!\!\!\! \alpha <0.5 
\end{cases}= \begin{cases}
   O\left( K    \right) &\!\!\!\! \alpha > 1 \\
        O\left(\frac{K}{\Lambda}  + \sqrt{K+ \frac{K^2}{(\log \Lambda)^2}} \right)              & \!\!\!\!\alpha =1 \\
       O\left( \frac{K}{\Lambda}  + \sqrt{ K+\frac{K^2}{\Lambda^{2-2\alpha}}}\right)               & \!\!\!\!0.5\!<\!\alpha\! <\!1 \\
        O\left( \frac{K}{\Lambda}  + \sqrt{K + \frac{K^2\log \Lambda}{\Lambda}}\right)            & \!\!\!\!\alpha =0.5 \\
	O\left(\frac{K}{\Lambda}  + \sqrt{\frac{K^2}{\Lambda} +K} 	    \right) & \!\!\!\!\alpha <0.5.
\end{cases} \nonumber
\end{align} 
Consequently, the $E[l_{1}]$ is upper bounded as
\begin{align} \label{eq:Eluba} 
E[l_{1}]= \begin{cases}
   O\left( K    \right) & \alpha > 1 \\
%        O\left(\frac{K}{\log \Lambda} \right)              & \alpha =1 \\
        O\left( \sqrt{K+ \frac{K^2}{(\log \Lambda)^2}} \right)             & \alpha =1 \\
       O\left(  \sqrt{ K+\frac{K^2}{\Lambda^{2-2\alpha}}}\right)             & 0.5<\alpha <1 \\
        O\left(  \sqrt{K + \frac{K^2\log \Lambda}{\Lambda}}\right)            & \alpha =0.5 \\
	O\left(\frac{K}{\Lambda}  + \sqrt{\frac{K^2}{\Lambda} +K} 	    \right) & \alpha <0.5.
\end{cases} 
\end{align}
From \eqref{eq:gapt} and \eqref{eq:Eluba}, the upper bound on $\overline{T}(\gamma)$ is given by
\begin{equation} \label{eq:Tuba} 
\overline{T}(\gamma)= \begin{cases}
   O\left( T_{min} \Lambda    \right) & \alpha > 1 \\
%\textcolor{blue}{        O\left( T_{min} \frac{\Lambda}{K} \sqrt{K+ \frac{K^2}{(\log \Lambda)^2}} \right)}              & \alpha =1 \\
        O\left( T_{min}  \sqrt{\frac{\Lambda^2}{K}+ \frac{\Lambda^2}{(\log \Lambda)^2}} \right)              & \alpha =1 \\
    %    O\left( T_{min} \frac{\Lambda}{\log \Lambda} \right)              & \alpha =1 \\
       O\left(T_{min}   \sqrt{ \frac{\Lambda^2}{K}+\Lambda^{2\alpha}}\right)             & 0.5<\alpha <1 \\
        O\left( T_{min}  \sqrt{\frac{\Lambda^2}{K} + \Lambda\log\Lambda}\right)            & \alpha =0.5 \\
	O\left(T_{min} \left(1  +  \sqrt{\Lambda+ \frac{\Lambda^2}{K} } \right) 	    \right) & \alpha <0.5.
\end{cases}
\end{equation}
Next, we characterize the lower bound on $\overline{T}(\gamma)$. We know that
\begin{align}
E[l_{1}]\geq K\max(\mathbf{p})= Kp_1= \frac{K}{H_{\alpha}(\Lambda)}.
\end{align}
Thus, from \eqref{eq:H} we have
\begin{equation} \label{eq:Ellba}
E[l_{1}]=   \begin{cases}
    \Omega\left(K \right) & \alpha > 1 \\
    \Omega\left(\frac{K}{\log \Lambda} \right) & \alpha = 1 \\
	\Omega\left(\frac{K}{\Lambda^{1-\alpha}}   \right)  & \alpha <1.
\end{cases} 
\end{equation}
From \eqref{eq:gapt} and \eqref{eq:Ellba}, the lower bound on $\overline{T}(\gamma)$ is given by
\begin{equation} \label{eq:Tlba}
\overline{T}(\gamma)= \begin{cases}
    \Omega\left(T_{min}  \Lambda \right) & \alpha > 1 \\
    \Omega\left(T_{min} \frac{ \Lambda}{\log \Lambda} \right) & \alpha = 1 \\
	\Omega\left(T_{min} \Lambda^{\alpha}   \right)  & \alpha <1.
\end{cases} 
\end{equation}
% \begin{equation} \label{eq:Tlba}
% \overline{T}(\gamma)= \begin{cases}
%     \Omega\left(T_{min}  \Lambda \right) & \alpha > 1 \\
%     \Omega\left(T_{min} \frac{ \Lambda}{\log \Lambda} \right) & \alpha = 1 \\
%     \Omega\left(T_{min} \Lambda^{\alpha} \right) & 0.5<\alpha < 1 \\
%     \Omega\left(T_{min} \sqrt{\Lambda}  \right)              & \alpha =0.5 \\
% 		\Omega\left(T_{min} \Lambda^{\alpha}   \right)  & \alpha <0.5.
% \end{cases} 
% \end{equation}
Combining \eqref{eq:Tuba} with \eqref{eq:Tlba}, allows us to directly conclude the proof of Theorem \ref{th:gapnn}.

\subsection{Proof of Theorem \ref{th:gapdnn}} \label{AP:gapdnn}
We proceed with the following lemma which is adopted and adapted here directly from the work of~\cite{wieder_07} on the \emph{Balls into Bins problem}.
\begin{lemma}[\!\!\!\protect{\cite[Theorem 1.3]{wieder_07} - adaptation}]\label{le:ballbinnn}
In a $\Lambda$-cell, $K$-user setting where each ball is associated to the least loaded  cache among $h \geq 2$ caches, independently sampled from $\Lambda$ caches whose population intensities vector is following the distribution $\mathbf{p}$, the tail of $l_1$ takes the form
\begin{align} \label{eq:acdfl1nn}
P[l_1 > \delta]=  o\left( \frac{1}{\Lambda}\right)
\end{align} 
for 
\begin{align} \label{eq:delta}
\!\!\delta\!= \frac{K}{\Lambda} + \log \log \Lambda + O\left(1\right),
\end{align}
when $h = \Theta \left(\frac{\log\left(\frac{ab-1}{a-1}\right)}{\log\left(\frac{ab-1}{ab-b}\right)}\right)$, where $a=\frac{1}{\Lambda \min(\mathbf{p})}$ and $b=\max(\mathbf{p}) \Lambda$.
\end{lemma}
\begin{proof}
The result comes directly from \cite[Theorem 1.3]{wieder_07}.
\end{proof}
With Lemma~\ref{le:ballbinnn} at hand, we get that 
\begin{align} \label{eq:ubl1nn}
E[l_1] &= \! \sum_{j=0}^{\delta-1}\!\! P[l_1 > j]\! +\!  P[l_1 > \delta] +\sum_{j=\delta+1}^{K-1}\!\! P[l_1 > j] \stackrel{(a)}{\leq} \delta +  o\left( \frac{1}{\Lambda}\right) +\sum_{j=\delta+1}^{K-1} P[l_1 > j] \nonumber \\
&\stackrel{(b)}{\leq} \delta +  o\left( \frac{1}{\Lambda}\right) + (K\!-\!\delta\!-\!1)o\left( \frac{1}{\Lambda}\right) = \delta \left(\!1\!-\!o\left( \frac{1}{\Lambda}\right)\!\!\right)\! +\! o\left( \!\frac{K}{\Lambda}\!\right) = O\left( \delta \right) +o\left( \frac{K}{\Lambda}\right),
\end{align}
where in step (a), we use the fact that $P[l_1 > j]$ is at most $1$ for $j= \left[0, 1, \cdots,  \delta-1\right]$ and in step (b), we use the fact that if $P[l_1 > \delta] =o(1)$ then $P[l_1 > j]$ is at most $o(1)$  for $j= \left[\delta+1, \delta+2, \cdots K-1\right]$.
Combining~\eqref{eq:delta} and~\eqref{eq:ubl1nn}, we have
\begin{align}\label{eq:AL1nn}
\!\!E[l_1]\! =O\left( \frac{K}{\Lambda} + \log \log \Lambda \right),
\end{align}
when $h = \Theta\left(\frac{\log\left(\frac{ab-1}{a-1}\right)}{\log\left(\frac{ab-1}{ab-b}\right)}\right)$. We now proceed to simplify $h$. When the cache population intensities $\mathbf{p}$ follows the Zipf distribution, we know that  $\min(\mathbf{p})= p_{\Lambda}= \frac{\Lambda^{-\alpha}}{H_{\alpha}(\Lambda)}$ and $\max(\mathbf{p})=p_1 = \frac{1}{H_{\alpha}(\Lambda)}$. We get the scaling of $h$ as
\begin{align}
h &= \Theta\left(\frac{\log\left(\frac{\frac{H_{\alpha}(\Lambda)}{\Lambda^{1-\alpha}}\frac{\Lambda}{H_{\alpha}(\Lambda)}-1}{\frac{H_{\alpha}(\Lambda)}{\Lambda^{1-\alpha}}-1}\right)}{\log\left(\frac{\frac{H_{\alpha}(\Lambda)}{\Lambda^{1-\alpha}}\frac{\Lambda}{H_{\alpha}(\Lambda)}-1}{\frac{H_{\alpha}(\Lambda)}{\Lambda^{1-\alpha}}\frac{\Lambda}{H_{\alpha}(\Lambda)}-\frac{\Lambda}{H_{\alpha}(\Lambda)}}\right)}\right) =    \Theta \left(\frac{\log\left(\frac{\Lambda(1-\Lambda^{-\alpha})}{H_{\alpha}(\Lambda)-\Lambda^{1-\alpha}}\right)}{\log\left(\frac{(1-\Lambda^{-\alpha})H_{\alpha}(\Lambda)}{H_{\alpha}(\Lambda)-\Lambda^{1-\alpha}}\right)}\right)\nonumber\\
&=    \Theta \left(\frac{\log\left(\frac{\Lambda(\Lambda^{\alpha}-1)}{\Lambda^{\alpha}H_{\alpha}(\Lambda)-\Lambda}\right)}{\log\left(\frac{(\Lambda^{\alpha}-1)H_{\alpha}(\Lambda)}{\Lambda^{\alpha}H_{\alpha}(\Lambda)-\Lambda}\right)}\right) = \Theta \left(\frac{ \log\Lambda+ \log\left(\frac{\Lambda^{\alpha}-1}{\Lambda^{\alpha}H_{\alpha}(\Lambda)-\Lambda}\right)}{\log\left(\frac{(\Lambda^{\alpha}-1)H_{\alpha}(\Lambda)}{\Lambda^{\alpha}H_{\alpha}(\Lambda)-\Lambda}\right)}\right)  .
\end{align}
Let $X=\frac{\Lambda^{\alpha}-1}{\Lambda^{\alpha}H_{\alpha}(\Lambda)-\Lambda}$, from \eqref{eq:H}, we get
\begin{equation} \label{eq:X} 
X = \begin{cases}
    \Theta\left(1 \right) & \alpha > 1 \\
   \Theta\left(\frac{1}{\log \Lambda} \right)              & \alpha =1 \\
	\Theta\left(\Lambda^{\alpha-1} \right)  & \alpha <1.
\end{cases} 
\end{equation}
As for $\alpha>1$, $H_{\alpha}(\Lambda)= \Theta\left(1 \right)$, thus  $\Lambda^{\alpha}H_{\alpha}(\Lambda)-\Lambda = \Theta\left(\Lambda^{\alpha} \right)$ and we get $X=\Theta\left(\frac{\Lambda^{\alpha}}{\Lambda^{\alpha} }\right)$. When $\alpha=1$, we have $H_{\alpha}(\Lambda)= \Theta\left(\log \Lambda \right)$, thus $\Lambda H_{\alpha}(\Lambda)-\Lambda = \Theta\left(\Lambda  \log \Lambda \right)$, which gives $X=\Theta\left(\frac{\Lambda}{ \Lambda  \log \Lambda}\right)$. When $\alpha<1$, we have  $H_{\alpha}(\Lambda)= \Theta\left( \Lambda^{1-\alpha} \right)$, thus  $\Lambda^{\alpha}H_{\alpha}(\Lambda)-\Lambda = \Theta\left(\Lambda\right)$  and we get $X=\Theta\left(\frac{\Lambda^{\alpha}}{ \Lambda}\right)$.
Consequently, we have
\begin{align}
h &=  \Theta \left(\frac{ \log\Lambda+ \log X}{\log\left(X H_{\alpha}(\Lambda)\right)}\right) \stackrel{(a)}{=} \Theta \left( \log\Lambda\!+\! \log X \right) =\begin{cases}
    \Theta\left(\log\Lambda \right) & \!\!\!\!\alpha > 1 \\
   \Theta\left(\log\Lambda \!+\! \log\left(\frac{1}{\log \Lambda}\right)\!\! \right)    =\Theta\left(\log\Lambda \right)          &\!\!\!\! \alpha =1 \\
	\Theta\left(\log\Lambda\!+ \! \log\left(\Lambda^{\alpha-1} \right)\!\right) =\Theta\left(\log\Lambda \right) & \!\!\!\!\alpha <1,
\end{cases} 
\end{align}
where in step (a), we use the fact that $XH_{\alpha}(\Lambda) = \Theta\left(1 \right)$, which is obtain by combining \eqref{eq:H} and \eqref{eq:X}. Finally, combining~\eqref{eq:gapt} with~\eqref{eq:AL1nn}, allows us to directly conclude the proof of Theorem~\ref{th:gapdnn}.

\bibliographystyle{IEEEtran}
\bibliography{IEEEabrv,mainJ}

\ifCLASSOPTIONcaptionsoff
  \newpage
\fi



%\bibliographystyle{IEEEtranTCOM}
%\begin{thebibliography}{1}


%\end{thebibliography}

\end{document}